%% file: main.tex
\documentclass[sigconf,screen]{acmart}

\acmConference[ESEC/FSE 2021]{The 29th ACM Joint European Software Engineering
Conference and Symposium on the Foundations of Software Engineering}{23 - 27
August, 2021}{Athens, Greece}

\usepackage{float}
\usepackage{amsmath,amsfonts}
\usepackage[ruled, vlined]{algorithm2e}
\usepackage{graphicx}
\usepackage{textcomp}
\usepackage{xcolor}
\usepackage{listings}
\usepackage{caption}
\usepackage{subcaption}
\usepackage{multirow}
\usepackage{booktabs}
\usepackage{makecell}
\usepackage{galois}
\usepackage{mathpartir}
\usepackage{bussproofs}
\usepackage{mathtools}
\usepackage{colortbl}
\usepackage{hhline}
\usepackage{stmaryrd}
\usepackage{microtype}

\input{macros}

\setcopyright{acmcopyright}
\acmPrice{15.00}
\acmDOI{10.1145/3468264.3468556}
\acmYear{2021}
\copyrightyear{2021}
\acmSubmissionID{fse21main-p162-p}
\acmISBN{978-1-4503-8562-6/21/08}
\acmConference[ESEC/FSE '21]{Proceedings of the 29th ACM Joint European Software Engineering Conference and Symposium on the Foundations of Software Engineering}{August 23--28, 2021}{Athens, Greece}
\acmBooktitle{Proceedings of the 29th ACM Joint European Software Engineering Conference and Symposium on the Foundations of Software Engineering (ESEC/FSE '21), August 23--28, 2021, Athens, Greece}

\begin{document}

\title{Accelerating JavaScript Static Analysis via Dynamic Shortcuts (Extended Version)}

\author{Joonyoung Park}
\authornote{Both authors contributed equally to the paper.}
\affiliation{%
  \institution{Korea Advanced Institute of Science and Technology}
  \state{Daejeon}
  \country{South Korea}
}
\email{gmb55@kaist.ac.kr}

\author{Jihyeok Park}
\authornotemark[1]
\affiliation{%
  \institution{Korea Advanced Institute of Science and Technology}
  \state{Daejeon}
  \country{South Korea}
}
\email{jhpark0223@kaist.ac.kr}

\author{Dongjun Youn}
\affiliation{%
  \institution{Korea Advanced Institute of Science and Technology}
  \city{Daejeon}
  \country{South Korea}
}
\email{f52985@kaist.ac.kr}

\author{Sukyoung Ryu}
\affiliation{%
  \institution{Korea Advanced Institute of Science and Technology}
  \city{Daejeon}
  \country{South Korea}
}
\email{sryu.cs@kaist.ac.kr}

\renewcommand{\shortauthors}{Park and Park, et al.}

\input{abstract}



\maketitle

\input{intro}
\input{motivation}
\input{shortcut}
\input{ds}
\input{javascript}
\input{implementation}
\input{eval}
\input{eval-opaque}
\input{related}
\input{conclusion}

\section*{Acknowledgements}
This work was supported by National Research Foundation of
Korea (NRF) (Grants NRF-2017R1A2B3012020 and 2017M3C4A7068177).

\balance
\bibliographystyle{ACM-Reference-Format}
\bibliography{ref}

\end{document}

%% file: abstract.tex
\begin{abstract}
JavaScript has become one of the most widely used programming languages for web
development, server-side programming, and even micro-controllers for IoT.
However, its extremely functional and dynamic features degrade the performance
and precision of static analysis.  Moreover, the variety of built-in functions
and host environments requires excessive manual modeling of their behaviors.  To
alleviate these problems, researchers have proposed various ways to leverage
dynamic analysis during JavaScript static analysis.  However, they do not fully
utilize the high performance of dynamic analysis and often sacrifice the
soundness of static analysis.

In this paper, we present \textit{dynamic shortcuts}, a new technique to
flexibly switch between abstract and concrete execution during JavaScript static
analysis in a sound way.  It can significantly improve the analysis performance
and precision by using highly-optimized commercial JavaScript engines and lessen
the modeling efforts for opaque code.  We actualize the technique via $\tool$,
an extended combination of SAFE and Jalangi, a static analyzer and a dynamic
analyzer, respectively.  We evaluated $\tool$ using 269 official tests of Lodash
4 library.  Our experiment shows that $\tool$ is 7.81$\x$ faster than the
baseline static analyzer, and it improves the precision to reduce failed
assertions by 12.31\% on average for 22 opaque functions.
\end{abstract}

%% file: intro.tex
\section{Introduction}\label{sec:intro}
Over the past decades, the rise of JavaScript as the de facto
language for web development has expanded its reach to diverse fields.
Node.js~\cite{nodejs} supports server-side programming, React
Native~\cite{react-native} and Electron~\cite{electron} produce cross-platform
applications, and Moddable~\cite{moddable} and Espruino~\cite{espruino}
provide JavaScript environments in micro-controllers for IoT.  Such wide prevalent
uses place JavaScript at \#7 programming language in the TIOBE Programming Community
index\footnote{https://www.tiobe.com/tiobe-index/}.  Thus, researchers have
developed static analyzers such as JSAI~\cite{jsai}, TAJS~\cite{tajs},
WALA~\cite{correlation}, and SAFE~\cite{safe,safe2} to understand behaviors of
JavaScript programs and to detect their bugs in a sound manner.

However, static analysis of real-world JavaScript programs suffers from immensely
functional and dynamic features of JavaScript such as callback
functions, first-class property names, and dynamic code generation.
While they provide flexibility in software development, it is
challenging to statically analyze such features.  To overcome these problems,
researchers have proposed several analysis techniques: advanced string
domains~\cite{string, regex, combining-string}, loop sensitivity~\cite{lsaECOOP,
lsaSPE}, analysis based on property relations~\cite{correlation, weaklyAPLAS,
weaklySPE, value-partitioning}, and on-demand backward
analysis~\cite{value-refinement}.

At the same time, JavaScript host environments require excessive
manual modeling of their behaviors for static analysis.  Because built-in functions and
host-dependent functions are implemented in native languages like C and
C++ instead of JavaScript, their code is \textit{opaque} during static
analysis.  Thus, static analyzers often model their behaviors
manually, which is error-prone, tedious, and labor-intensive.
While researchers have proposed automatic modeling
techniques~\cite{safewapi, safets}, since they utilize only type
information, they generate imprecise models compared with the manual approach.

\begin{figure}[t]
  \centering
  \vspace{2mm}
  \includegraphics[width=\linewidth]{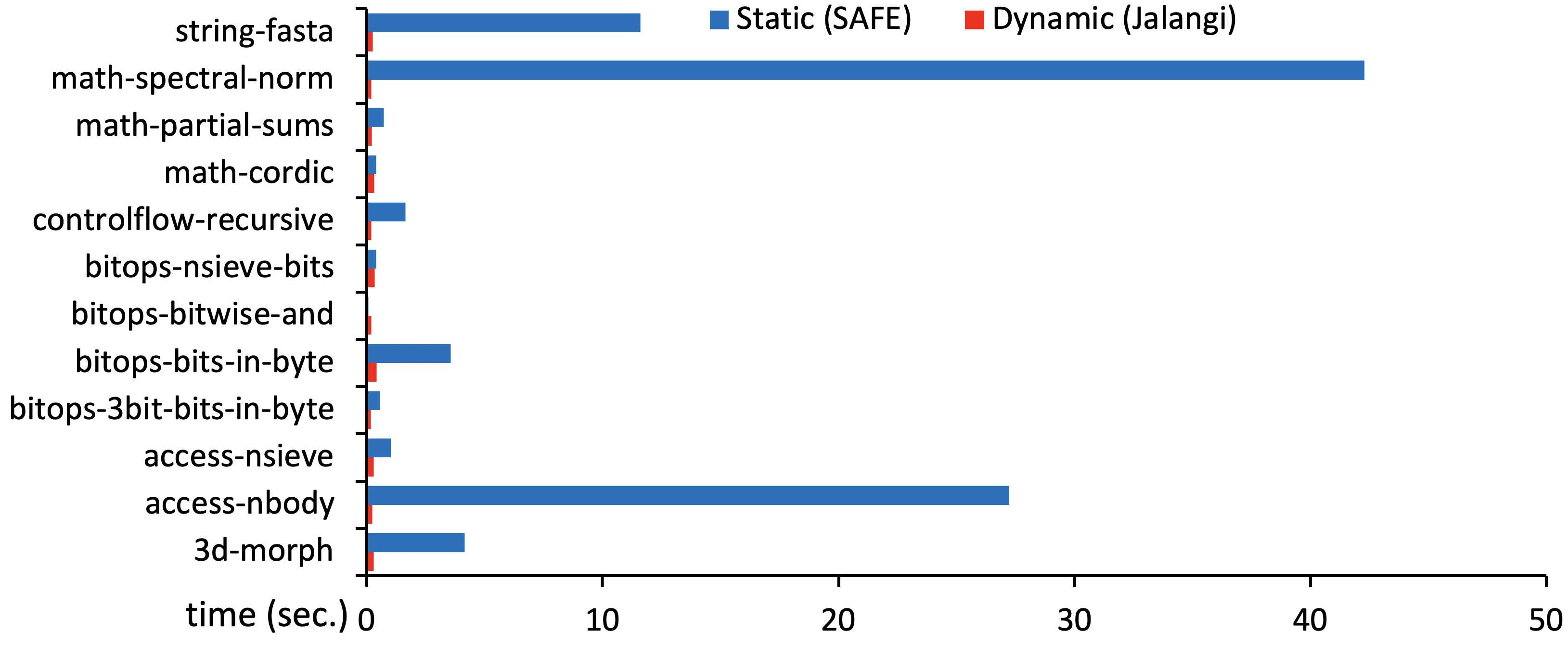}
  \vspace*{-1.5em}
  \caption{Performance of a dynamic analyzer and a static analyzer for a subset
  of the SunSpider benchmark}
  \label{fig:performance}
\end{figure}

To alleviate these problems, researchers have leveraged
dynamic analysis during static analysis.  Unlike static analyzers that
run on their own interpreters, dynamic analyzers such as Jalangi~\cite{jalangi} and DLint~\cite{dlint}
run on highly-optimized commercial JavaScript engines, which makes them
much faster than static analyzers.
Figure~\ref{fig:performance} shows that the dynamic analyzer
Jalangi is 34.8x faster than the static analyzer SAFE for a subset of the
SunSpider~\cite{sunspider} benchmark that is input-independent and deterministic.
Using high performance dynamic analysis, researchers have
reduced the scope of static analysis~\cite{determinacy, blendedJS} and constructed
initial abstract states~\cite{battles, eha} and automatic modeling of opaque
code~\cite{sra}.

Unfortunately, existing techniques using dynamic analysis for static analysis
have two limitations: 1) they do not fully utilize the high performance of dynamic
analysis, and 2) they sacrifice the soundness of static analysis.  Most of
them are \textit{staged analyses}, which first extract specific
information via dynamic analysis and utilize it in static analysis.
\citet{determinacy} identify determinate expressions that always have the same
values at given program points, \citet{blendedJS} extract dynamic values to
change expressions to certain literals, and Park et al.~\cite{battles,eha} dump
the initial states of a certain host environment or the entry of an event handler.
However, because they do not utilize dynamic analysis as soon as static
analysis begins, they do not get performance benefits since then.  Moreover,
they sacrifice the soundness of static analysis by performing dynamic analysis.
For example, the SRA model~\cite{sra} uses dynamic analysis for opaque code
with abstract arguments during static analysis.  When the abstract arguments
represent an infinite number of values, it randomly samples finite
concrete values for the abstract arguments, which makes the analysis
result unsound due to missing concrete values.

In this paper, we present \textit{dynamic shortcuts}, a new technique to
flexibly switch between abstract and concrete execution during JavaScript static
analysis in a sound way.  During static analysis, one can take a dynamic
shortcut, which consists of three parts: 1) converting the current abstract
state to its corresponding \textit{{\sealed} state}, 2) performing
\textit{{\sealed} execution} on the {\sealed} state, and 3) converting the
result of the {\sealed} execution to its corresponding abstract state.  Our key
observation is that we can use the fast concrete execution for specific program
parts while preserving the soundness if they do not use abstract values.  For
example, consider static analysis of the following JavaScript code:
\begin{lstlisting}[style=myJSstyle,numbers=none]
    var v = ... // an abstract value
    var obj = { p1: v }, y = "p";
    x = obj[y + 1];
\end{lstlisting}
Because \jscode{y} stores a string \jscode{"p"}, the expression \jscode{y + 1}
evaluates to a string \jscode{"p1"} and \jscode{x = obj[y + 1]} assigns the
abstract value of \jscode{v} stored in \jscode{obj.p1} to the variable
\jscode{x}.  Note that even though \jscode{obj} contains an abstract value
\jscode{v}, because the third line does not ``use'' the value of \jscode{v} but
only ``passes'' it to the variable \jscode{x}, we can concretely execute the code.
Based on this observation, we introduce {\sealed} execution, which is concrete
execution using \textit{{\sealed} values}.  A {\sealed} value is a symbol that
represents an abstract value in {\sealed} execution; it signals the end of the
current dynamic shortcut when the {\sealed} execution tries to access its value.
To evaluate our technique, we implemented $\tool$ using SAFE and Jalangi and
analyzed 269 official tests of Lodash 4 library.

The contributions of this paper include the following:
\begin{itemize}
\item We present a novel technique for JavaScript static
analysis to leverage the high performance of dynamic analysis using
dynamic shortcuts.  We formally define the technique and prove
its soundness and termination.
\item We actualize the proposed technique in $\tool$, an
extended combination of SAFE and Jalangi.
\item For empirical evaluation, we analyzed 269 official tests of
Lodash 4 library.  The experiment shows that $\tool$ outperforms
SAFE 7.81$\x$ on average.  Moreover, by using dynamic shortcuts
instead of manual modeling for 22 opaque functions,
$\tool$ improves the analysis precision to reduce failed assertions by 12.31\% on
average.
\end{itemize}

In the remainder of this paper, Section~\ref{sec:motivation} explains the
motivation of this work with a simple example.  Section~\ref{sec:formal}
formalizes the language-agnostic part of the technique in the abstract
interpretation framework.  Then, we extend the formalization with
JavaScript specific features in Section~\ref{sec:javascript}.
Section~\ref{sec:implementation} describes important details of the
$\tool$ implementation.  We explain the evaluation results of $\tool$
with real-world benchmarks in Section~\ref{sec:eval}.
Section~\ref{sec:related} discusses related work and
Section~\ref{sec:conclusion} concludes.

%% file: motivation.tex
\section{Motivation}\label{sec:motivation}

\begin{figure}[t]
  \centering
  \vspace{2mm}
  \begin{subfigure}[t]{0.5\textwidth}
    \begin{lstlisting}[style=myJSstyle,xleftmargin=2.5em]
function concat() {
  var length = arguments.length;
  if (!length) return [];
  var array = arguments[0],
      args  = Array(length - 1),
      index = length;
  while (index--)
    args[index-1] = arguments[index];
  return arrayPush(isArray(array) ?
    copyArray(array) : [array],
    baseFlatten(args, 1));
}
    \end{lstlisting}
    \vspace*{-.5em}
    \caption{Lodash's \jscode{concat} function}
    \label{fig:concat}
  \end{subfigure}
  \begin{subfigure}[t]{0.5\textwidth}
    \begin{lstlisting}[style=myJSstyle,firstnumber=13,xleftmargin=2.5em]
function changeCountry(G) { ...
  if (G.selectedVal === "US" && state) {
    // deterministic arguments of `concat`
    state.items = _.concat([["Other", "Other"]],
      WebinarBase.questions.state.items);
    state.selectedVal = _.head(_.head(C.items));
  }
}
    \end{lstlisting}
    \vspace*{-.5em}
    \caption{Call of \jscode{concat} with concrete values}
    \label{fig:changeCountry}
  \end{subfigure}
  \begin{subfigure}[t]{0.5\textwidth}
    \begin{lstlisting}[style=myJSstyle,firstnumber=22,xleftmargin=2.5em]
function getData(e) {
  var option = ... // option for server connection
  post(option).then(function(e) {
    if (e.total_records && e.total_records > 0) {
      // non-deterministic arguments of `concat`
      this.pastEvents =
        _.concat(this.pastEvents, e.events);
      this.total = e.total_records;
    } else this.noPastData = !0
  })
}
    \end{lstlisting}
    \vspace*{-.5em}
    \caption{Call of \jscode{concat} with abstract values}
    \label{fig:getData}
  \end{subfigure}
  \vspace*{-.5em}
  \caption{Lodash library function and its uses in \code{zoom.us}}
  \label{fig:example}
\end{figure}

This section explains the motivation of dynamic shortcuts
using real-world examples in Figure~\ref{fig:example}.
We describe their behaviors and explain how we can utilize
dynamic shortcuts during static analysis.

Figure~\ref{fig:example}(a) shows the \jscode{concat} function defined
in Lodash library~\cite{lodash} (v4.17.20); it is
the most popular npm package\footnote{https://www.npmjs.com/browse/depended}
and 131,517 npm packages have a dependency on it.
The \jscode{concat} function creates a new array concatenating given arrays or values.
It first checks the length of \jscode{arguments} on lines 1--3.
Then, it stores the first argument to \jscode{array} on line 4 and
copies the remaining arguments to \jscode{args} on lines 5--8.
On line 9, it checks whether \jscode{array} is an array object using
the built-in function \jscode{isArray}.  If so, it creates a new array
by copying the given array via \jscode{copyArray}; otherwise,
it creates a singleton array \jscode{[array]}.  Finally, it flattens
\jscode{args} via \jscode{baseFlatten} and pushes the result to the
new array on line 11.

Figure~\ref{fig:example}(b) and Figure~\ref{fig:example}(c) show use
cases of the \jscode{concat} function in the \code{zoom.us}~\cite{zoom} site.
It is the homepage of Zoom, a videotelephony software by Zoom
Video Communications and it is ranked as the 15th popular web site according
to Alexa\footnote{https://www.alexa.com/siteinfo/zoom.us} in February 2021.

\paragraph{Dynamic shortcuts with concrete values.}
When a function is called with concrete values, we can perform
dynamic analysis instead of static analysis.
For example, \jscode{changeCountry} in Figure~\ref{fig:example}(b)
is invoked when a user selects a country from a drop-down list in the registration page.
It calls the \jscode{concat} function to update the drop-down list of
states or provinces on lines 16--17.  However, when the user selects ``United States of America,''
which is \jscode{"US"}, two arguments are pre-defined with
deterministic values; the first one is an array literal
\jscode{[["Other", "Other"]]} and the second one is an array of pairs
of abbreviations and names of the states defined as follows: \\[-.5em]
\begin{lstlisting}[style=myJSstyle,numbers=none]
WebinarBase.questions.state.items =
  [["AL","Alabama"], ..., ["WY", "Wyoming"]]
\end{lstlisting}
\noindent\smallskip
Moreover, \jscode{this} also has a concrete value, the Lodash top-level object~\jscode{\_}.
Thus, we can perform dynamic analysis by invoking \jscode{concat} with~\jscode{\_}
as its \jscode{this} value and the above concrete values as arguments. 
By skipping the analysis of the function call on lines 17--18 and
utilizing the result of dynamic analysis, it improved the analysis performance.

\paragraph{Dynamic shortcuts with abstract values.}

Even when a function is called with abstract values, we can still
perform dynamic analysis using {\sealed} execution.
For example, \jscode{getData} in Figure~\ref{fig:example}(c) is invoked
when a user clicks the ``Load More'' button to load more Zoom events in
the ``Webinars \& Events'' page.  It sends a POST
request to a server and receives additional events \jscode{e} on
line 24.  Then, eight events in \jscode{e.events} are appended to
\jscode{this.pastEvents} using \jscode{concat} on lines 27--28.
However, the arguments of \jscode{concat} are not deterministic because 1) the
event list stored in \jscode{this.pastEvents} is continuously grown for each
load and 2) the events stored in \jscode{e.events} are dependent on the
data given from the server.

\begin{figure}[t]
  \vspace{2mm}
  \begin{subfigure}{0.23\textwidth}
    \[
      \begin{array}{|c|c|}\hline
        \text{Property} & \text{Value}\\\hline
        \top & \symb_\jscode{evt}\\\hline
        \jscode{"length"} & \symb_\jscode{int}\\\hline
      \end{array}
    \]
    \caption{\jscode{this.pastEvents}}
    \label{fig:pastEvents}
  \end{subfigure}
  \begin{subfigure}{0.23\textwidth}
    \[
      \begin{array}{|c|c|}\hline
        \text{Property} & \text{Value}\\\hline
        \jscode{0} & \symb_\jscode{evt}\\\hline
        \cdots & \cdots\\\hline
        \jscode{7} & \symb_\jscode{evt}\\\hline
        \jscode{"length"} & \jscode{8}\\\hline
      \end{array}
    \]
    \caption{\jscode{e.events}}
    \label{fig:events}
  \end{subfigure}
  \vspace*{-.5em}
  \caption{Concrete objects with {\sealed} values}
  \vspace*{-1em}
  \label{fig:sealed}
\end{figure}

To perform dynamic analysis with abstract values, we seal abstract values
with {\sealed} values as in Figure~\ref{fig:sealed}.  Two {\sealed}
values $\symbevt$ and $\symbint$ represent an event
object and an integer, respectively.  Then, we can perform dynamic analysis
successfully until line 9.  On line 2, \jscode{length} is \jscode{2};
on line 4, \jscode{array} points to \jscode{this.pastEvents};
on lines 5--8, \jscode{args} stores an array with a single object stored in \jscode{e.events};
and on line 9, \jscode{isArray(array)} is \jscode{true}.
However, dynamic analysis fails for \jscode{copyArray(array)} on line 10
because the value of the \jscode{length} property of \jscode{array} is the {\sealed} value $\symbint$.
Then, we stop the {\sealed} execution, convert the current
{\sealed} state to its corresponding abstract state, and resume
the static analysis from line 10.  Because {\sealed} execution
leverages fast dynamic analysis as long as possible, the overall
analysis becomes more scalable.

\paragraph{Dynamic shortcuts for opaque functions.}
As the previous two examples additionally show, using dynamic shortcuts 
lessens the burden of modeling opaque functions from static analysis, and
it can even improve the analysis precision.
On line 9, since the \jscode{isArray} function is a JavaScript built-in library function,
it is implemented in a native language of the host environment, which
often requires manual modeling of its behaviors for JavaScript static analysis.
Assuming that a static analyzer models \jscode{isArray} to return the
boolean top value $\top_b$ that encompasses both \jscode{true} and \jscode{false},
static analysis of the ternary conditional expression on lines 9--10 analyzes
both branches \jscode{copyArray(array)} and \jscode{[array]}, even though
\jscode{[array]} is never reachable in the example code.  On the contrary,
using dynamic shortcuts, static analysis does not need to model \jscode{isArray}.
It can perform {\sealed} execution for \jscode{isArray}, which
returns a more precise result \jscode{true} than $\top_b$.

%% file: shortcut.tex
\section{Dynamic Shortcuts}\label{sec:formal}
In this section, we formally define static analysis using dynamic shortcuts by
introducing {\sealed} execution in the abstract interpretation framework.
We extend the formalization of abstract interpretation of \citet{abs-interp-1977,
abs-interp-1992} and views-based analysis sensitivity of \citet{sens-toplas}.
For dynamic shortcuts, we define {\sealed} execution with a
{\sealed} domain and abstract instantiation maps.  To combine
sensitive abstract interpretation and {\sealed} execution, we define
a combined domain of sensitive abstract domain and {\sealed} domain and
explain it with a simple example. Finally, we prove the soundness and
termination property of abstract interpretation using the combined domain.

\subsection{Concrete Semantics}

We define a program $\prog$ as a state transition system $(\stset, \trans,
\istset)$.  A program starts with an initial state in $\istset$ and the
transition relation $\trans \subseteq \stset \times \stset$ describes how states
are transformed to other states.  A \textit{collecting semantics} $\sem{\prog} =
\{ \st \in \stset \mid \ist~\in~\istset \wedge \ist \trans^* \st \}$ consists of
reachable states from initial states of the program $\prog$.  We can compute
it using a \textit{transfer function} $\transfer: \dom \rightarrow \dom$ as
follows:
\[
  \sem{\prog} = \underset{n \rightarrow \infty}{\lim}{\transfer^n(\ielem)}\\
  \qquad
  \transfer(\elem) = \elem \join \step(\elem)\\
\]
where the \textit{concrete domain} $\dom = \powerset{\stset}$ is a complete lattice
with $\cup$, $\cap$, and $\subseteq$ as its join($\join$), meet($\meet$), and
partial order($\order$) operators.  The set of states $\ielem$ denotes the
initial states $\istset$.  The \textit{one-step execution} $\step: \dom
\rightarrow \dom$ transforms states using the transition relation $\trans$:
$\step(\elem) = \{ \st' \mid \st \in \elem \wedge \st \trans \st' \}$.

\begin{figure}[t]
  \vspace{2mm}
  \[
    \begin{array}{r@{~}l@{~}c@{~}l}
      \labdot{0} & \kwif \; (\; \varx \geq 0 \;) & \labdot{1} & \varx = \varx ;\\
                 & \kwelse & \labdot{2} & \varx = -\varx ;\\
      \labdot{3} & \varx = -\varx ; & \labdot{4} \\
    \end{array}
  \]
  \vspace*{-.5em}
  \caption{Negation of the absolute value of $\varx$}
  \vspace*{-.5em}
  \label{fig:running-example}
\end{figure}

For example, the code in Figure~\ref{fig:running-example} is a simple program
that calculates the negation of the absolute value of the variable $\varx$.
States are pairs of labels and integers stored in $\varx$: $\stset = \labset
\times \mathbb{N}$.  Assume that the initial states are $\istset = \{ (\lab_0,
-42) \}$, which denotes that the program starts at $\lab_0$
with the variable $\varx$ of value $-42$.
Then, it executes with the following trace:
\[
  (\lab_0, -42) \trans (\lab_2, -42) \trans (\lab_3, 42) \trans (\lab_4, -42)
\]

\subsection{Abstract Interpretation}\label{sec:ai}
Abstract interpretation~\cite{abs-interp-1977, abs-interp-1992}
over-approximates the transfer function $\transfer$ as an \textit{abstract transfer
function} $\abstransfer: \absdom \rightarrow \absdom$ to get an
\textit{abstract semantics} $\abssem{\prog}$ in finite iterations as follows:
\[
    \abssem{\prog} = \underset{n \rightarrow
    \infty}{\lim}{(\abstransfer)^n(\iabselem)}\\
\]
We define a \textit{state abstraction} $\dom \galois{\alpha}{\gamma} \absdom$ as
a Galois connection between the concrete domain $\dom$ and an abstract domain
$\absdom$ with a \textit{concretization function} $\gamma$ and an
\textit{abstraction function} $\alpha$.  The initial abstract state $\iabselem
\in \absdom$ represents an abstraction of the initial state set: $\ielem
\subseteq \gamma(\iabselem)$.  The abstract transfer function $\abstransfer:
\absdom \rightarrow \absdom$ is defined as $\abstransfer(\abselem) = \abselem
\join \absstep(\abselem)$ with an \textit{abstract one-step execution}
$\absstep: \absdom \rightarrow \absdom$.  For a sound state abstraction, the
join operator and the abstract one-step execution should satisfy the following
conditions:
\begin{align}
  \forall \abselem_0, \abselem_1 \in \absdom & . \; \gamma(\abselem_0) \cup
  \gamma(\abselem_1) \subseteq \gamma(\abselem_0 \join
  \abselem_1) \label{equ:sound-join}\\
  \forall \abselem \in \absdom & . \; \step \circ \gamma(\abselem) \subseteq
  \gamma \circ \absstep(\abselem) \label{equ:sound-step}
\end{align}

A simple example abstract domain is $\absdom_\pm = \powerset{\{ -, +, 0 \}}$ with
set operators as domain operators; $-$ denotes negative integers, $+$ positive
integers, and $0$ zero.  Assume that we analyze the code in
Figure~\ref{fig:running-example} with the abstract domain and the initial abstract state $\iabselem =
\{ - \}$. Then, the analysis result is $\{ -, + \}$ because $\varx$ can
have a positive value by executing $\varx = -\varx$ but there is no
way for $\varx$ to have $0$ in this program.

\subsection{Analysis Sensitivity}\label{sec:sens}

Abstract interpretation is often defined with \textit{analysis sensitivity} to
increase the precision of static analysis.  A sensitive abstract domain
$\sabsdom: \viewset \rightarrow \absdom$ is defined with a \textit{view
abstraction} $\viewmap: \viewset \rightarrow \dom$ that provides multiple points
of views for reachable states during static analysis.  It maps a finite number
of views $\viewset$ to sets of states $\dom$. Each view $\view \in \viewset$
represents a set of states $\viewmap(\view)$ and each state is included
in a unique view: $\forall \st \in \stset. \; \st \in \viewmap(\view)
\Rightarrow \forall \view' \in \viewset. \st \in \viewmap(\view') \Rightarrow \view = \view'$.
A \textit{sensitive state
abstraction} $\dom \galois{\alpha_\viewmap}{\gamma_\viewmap} \sabsdom$ is a
Galois connection between the concrete domain $\dom$ and the sensitive abstract
domain $\sabsdom$ with the following concretization function:
\[
  \sgamma(\sabselem) = \underset{\view \in \viewset}{\bigcup}
  {\viewmap(\view) \cap \gamma \circ \sabselem(\view)}
\]

With analysis sensitivities, the abstract one-step execution $\sabsstep:
\sabsdom \rightarrow \sabsdom$ is defined as follows:
\[
  \sabsstep(\sabselem) = \lambda \view \in \viewset. \; \underset{\view' \in
  \viewset}{\bigjoin}{\viewtrans{\view'}{\view} \circ \sabselem(\view')}
\]
where $\viewtrans{\view'}{\view}: \absdom \rightarrow \absdom$ is an abstract
semantics of a \textit{view transition} from a view $\view'$ to another view
$\view$.  It should satisfy the following condition for the soundness of the
analysis:
\[
  \forall \abselem \in \absdom. \; \step(\gamma(\abselem) \cap \viewmap(\view'))
  \cap \viewmap(\view) \subseteq \gamma \circ
  \viewtrans{\view'}{\view}(\abselem)
\]

One of the most widely-used analysis sensitivity is \textit{flow sensitivity}
defined with a flow-sensitive view abstraction $\fsviewmap: \labset
\rightarrow \dom$ where:
\[
  \forall \lab\in\labset. \; \fsviewmap(\lab) = \{ \st \mid \st = (\lab, \_) \}
\]
If we apply the flow sensitivity for the above example with the initial abstract
state $[ \lab_0 \mapsto \{ -, 0, + \} ]$, the analysis result is as follows:
\vspace*{.5em}
\[
  \begin{array}{|c||c|c|c|c|c|}\hline
    \labset & \lab_0 & \lab_1 & \lab_2 & \lab_3 & \lab_4\\\hline
    \absdom_\pm & -, 0, + & 0, + & - & 0, + & -, 0\\\hline
  \end{array}
\]

\subsection{{\SealeD} Execution}

We define \textit{{\sealed} execution} by extending the transition
relation $\trans$ as a {\sealed} transition relation $\symbtrans$ on {\sealed}
states.  First, we extend concrete states $\stset$ to {\sealed} states
$\symbstset$ by extending values $\valset$ with \textit{{\sealed} values}
$\symbset$.  We also define the {\sealed} transition relation $\symbtrans
\subseteq \symbstset \times \symbstset$. We use the notation $\symbtrans^k$
for $k$ repetition of $\symbtrans$, and write $\symbst \symbtrans \excst$ when
$\symbst$ does not have any {\sealed} transitions to other {\sealed}
states.  We define the validity of {\sealed} execution as follows:
\begin{definition}[Validity]\label{def:valid-symbtrans}
  The {\sealed} transition relation is \textit{valid} when the following
  condition is satisfied for any {\sealed} states $\symbst$ and
  $\symbst'$:
  \[
    \symbst \symbtrans \symbst' \Leftrightarrow
    \forall \imap \in \imapset. \;
    \{ \st' \mid \instant{\symbst}{\imap} \trans \st' \}
    = \{ \instant{\symbst'}{\imap} \}
  \]
  where $\imapset: \symbset \rightarrow \valset$ represent \textit{instantiation
  maps} from {\sealed} values to concrete values, and $\instant{\symbst}{\imap}$
  denotes a state produced by replacing each {\sealed} value $\symb$ in
  $\symbst$ with its
  corresponding value $\imap(\symb)$ using the instantiation map $\imap \in
  \imapset$.
\end{definition}

{\Sealed} execution is different from traditional
symbolic execution~\cite{symbolic} in that it supports only {\sealed}
values instead of symbolic expressions and path constraints.  For example, the
following trace represents traditional symbolic execution of the running
example in Figure~\ref{fig:running-example}:
{
\small
\[
  \begin{array}{r@{~}c@{~}c@{~}c@{~}r@{~}c@{~}r}
    &&(\lab_1, \symb)[\symb \!\geq\! 0]
    &\trans& (\lab_3, \phantom{-}\symb)[\symb \!\geq\! 0]
    &\trans& (\lab_4, -\symb)[\symb \!\geq\! 0]
    \vspace*{-0.5em}\\
    &\rutrans&
    \vspace*{-0.5em}\\
    (\lab_0, \symb)[\varnothing]
    \vspace*{-0.5em}\\
    &\rdtrans&
    \vspace*{-0.5em}\\
    &&(\lab_2, \symb)[\symb \!<\! 0]
    &\trans& (\lab_3, -\symb)[\symb \!<\! 0]
    &\trans& (\lab_4, \phantom{-}\symb)[\symb \!<\! 0]\\
  \end{array}
\]
}
It first assigns a symbolic value $\symb$ to the variable $\varx$ at $\lab_0$.
For the conditional branch, it creates two symbolic states with
different path conditions $\symb \geq 0$ and $\symb < 0$ for true and false
branches, respectively.  After executing statements $\varx = \varx$ and $\varx =
-\varx$, the variable $\varx$ stores symbolic expressions $\symb$ and $-\symb$
at $\lab_3$, respectively. Similarly, $\varx$  stores $-\symb$ and $\symb$ at $\lab_4$.
However, {\sealed} execution stops at $\lab_0$ as follows:
\[
  (\lab_0, \symb) \; \symbtrans \; \excst
\]
because the branch requires the actual value of the {\sealed} value $\symb$.

To define an abstract domain that contains {\sealed} states, we define
\textit{abstract instantiation maps} $\absimapset: \symbset \rightarrow
\absvalset$ from {\sealed} values to abstract values.  Its concretization
function $\imapgamma: \absimapset \rightarrow \powerset{\imapset}$ is defined
with the concretization function $\valgamma: \absvalset \rightarrow
\powerset{\valset}$ for values as follows:
\[
  \imapgamma(\absimap) = \{
    \imap \mid \forall \symb \in \symbset. \;
    \imap(\symb) \in \gamma \circ \absimap(\symb)
  \}
\]
The instantiation of a given {\sealed} state $\symbst \in \symbstset$ with
an abstract instantiation map $\absimap \in \absimapset$ is defined as follows:
\[
  \instant{\symbst}{\absimap} = \{ \instant{\symbst}{\imap} \mid \imap \in
  \imapgamma(\absimap) \}
\]
Now, we define a \textit{{\sealed} domain} as follows:

\begin{definition}[{\SealeD} Domain]\label{def:symbdom}
  A \textit{{\sealed} domain} $\symbdom: \powerset{\absimapset \times
  \symbstset}$ is defined with the concretization function
  $\symbgamma: \symbdom \rightarrow \dom$ and the {\sealed} one-step execution
  $\symbstep: \symbdom \rightarrow \symbdom$ such that
  \begin{align}
    \symbgamma(\symbelem) &=
    \bigcup \{ \instant{\symbst}{\absimap} \mid (\absimap, \symbst) \in
    \symbelem\}\\
    \symbstep(\symbelem) &= \{ (\absimap, \symbst') \mid (\absimap, \symbst)
    \in \symbelem \wedge \symbst \symbtrans \symbst' \}
  \end{align}
\end{definition}

%% file: ds.tex
\subsection{Combined Domain}
We now define a \textit{combined domain} of a given sensitive abstract
domain with the {\sealed} domain and its one-step execution.
\begin{definition}[Combined Domain]
  A \textit{combined domain} is $\combdom = \sabsdom \times \symbdom$ and its
  concretization function $\combgamma: \combdom \rightarrow \dom$ and join
  operator are defined as follows:
  \begin{align}
    \combgamma((\sabselem, \symbelem)) &= \sgamma(\sabselem) \cup
      \symbgamma(\symbelem)\\
    (\sabselem, \symbelem) \join ({\sabselem}', \symbelem') &= (\sabselem \join
      {\sabselem}', \symbelem \cup \symbelem')
  \end{align}
\end{definition}

Before defining the one-step execution for the combined domain, we introduce
\textit{analysis elements} to easily configure different types of abstract
states in the sensitive abstract domain and the {\sealed} domain.
\begin{definition}[Analysis Elements]\label{def:aelem}
  An \textit{analysis element} $\aelem \in \aelemset = (\viewset \times \absdom)
  \uplus (\absimapset \times \symbstset)$ is either 1) a pair of a view and an
  abstract state in a sensitive abstract domain $\sabsdom$, or 2) a pair of an
  abstract instantiation map and a {\sealed} state in a {\sealed}
  domain $\symbdom$.  Its concretization function $\aelemgamma:
  \aelemset \rightarrow \dom$ is defined as follows:
  \[
    \aelemgamma(\aelem) = \left\{
      \begin{array}{ll}
        \viewmap(\view) \cap \gamma(\abselem) & \text{if} \; (\view, \abselem) = \aelem\\
        \instant{\symbst}{\absimap} & \text{if} \; (\absimap, \symbst) = \aelem\\
      \end{array}
    \right.
  \]
\end{definition}

\begin{figure*}[t]
  \vspace{2mm}
  \centering
  \begin{subfigure}[t]{0.15\textwidth}
    \includegraphics[height=3.2cm]{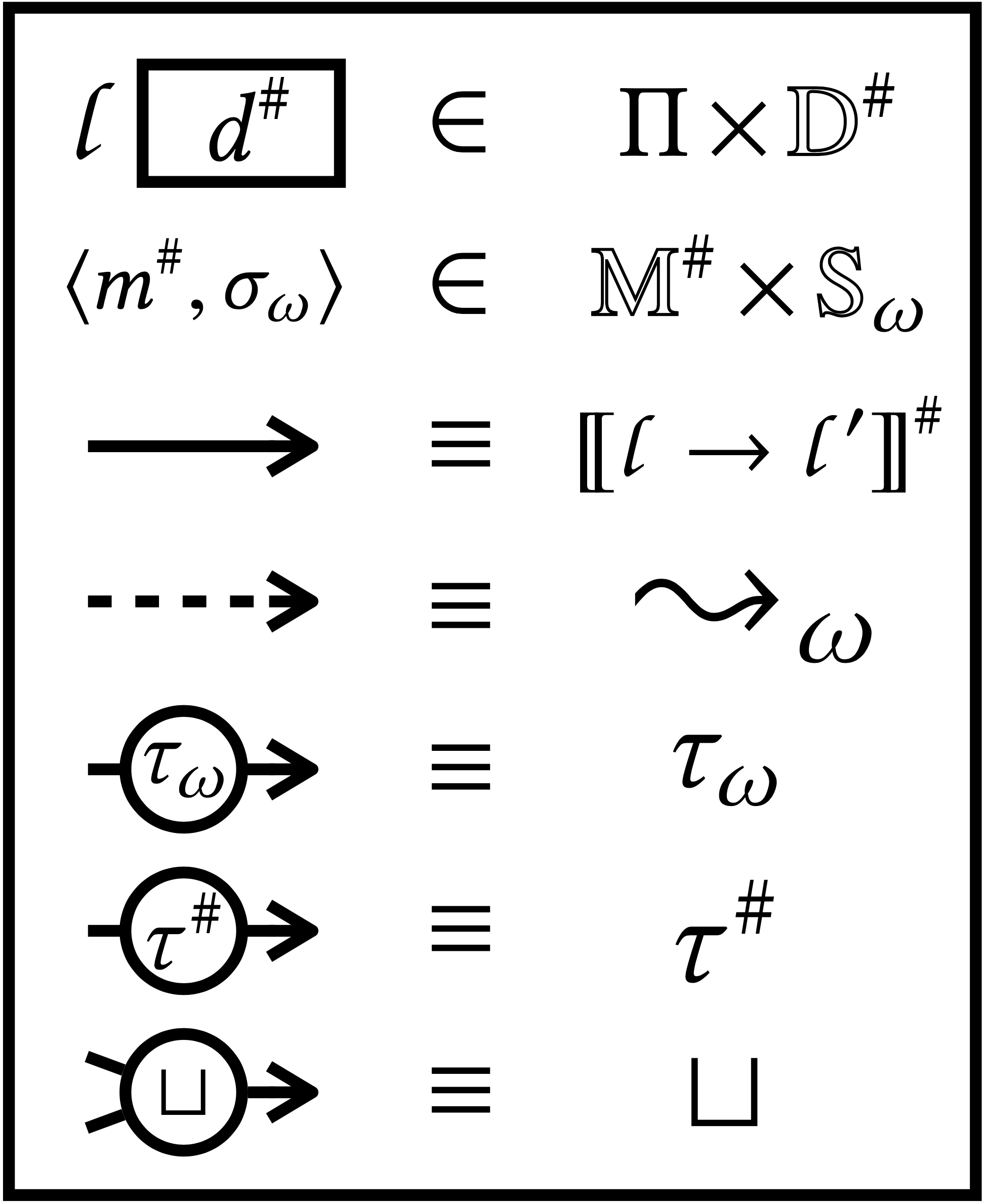}
    \caption{Notations}
    \label{fig:ds-example1}
  \end{subfigure}
  \quad
  \begin{subfigure}[t]{0.23\textwidth}
    \includegraphics[height=3.2cm]{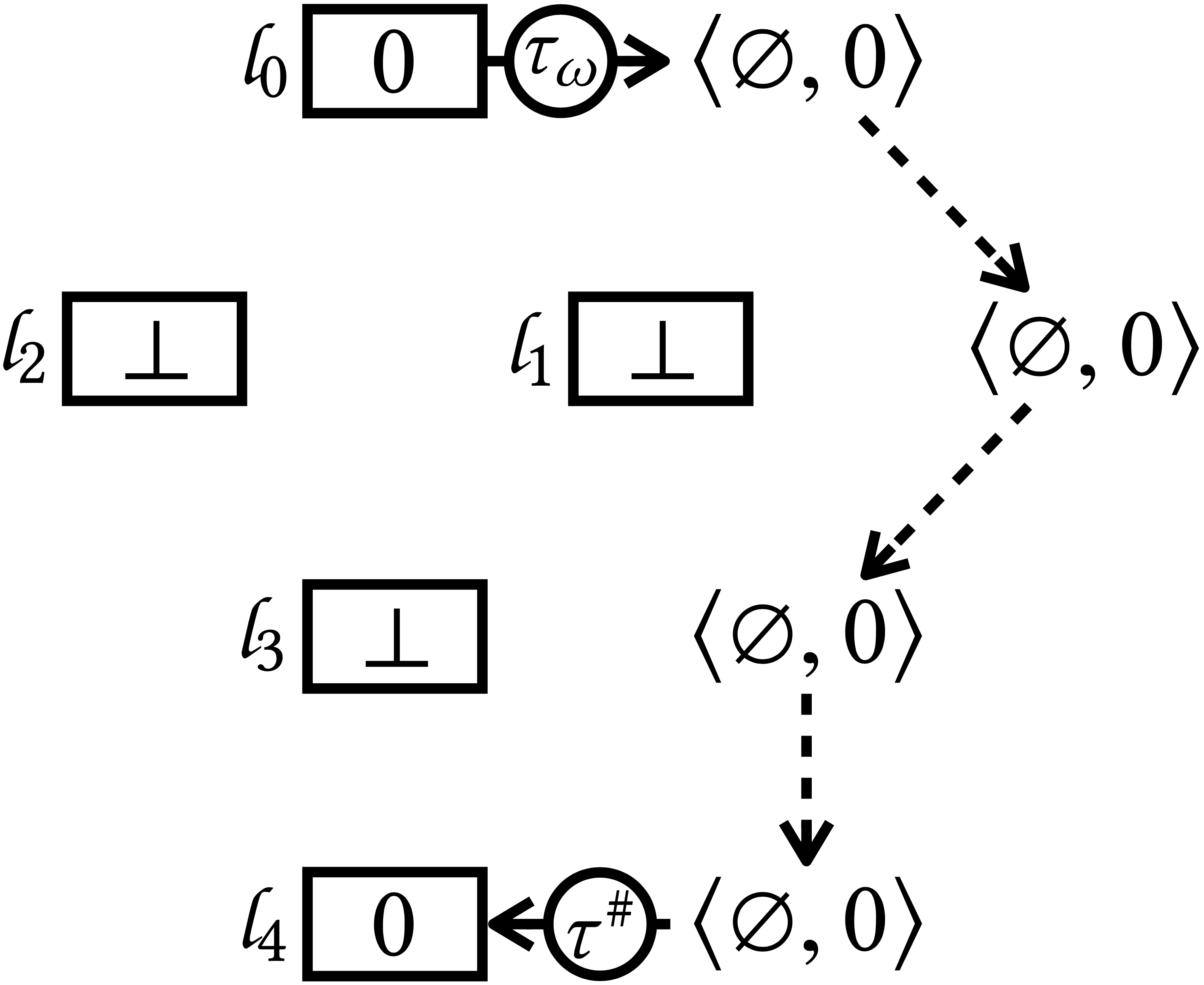}
    \caption{$\varx = 0$}
    \label{fig:ds-example2}
  \end{subfigure}
  \begin{subfigure}[t]{0.28\textwidth}
    \includegraphics[height=3.2cm]{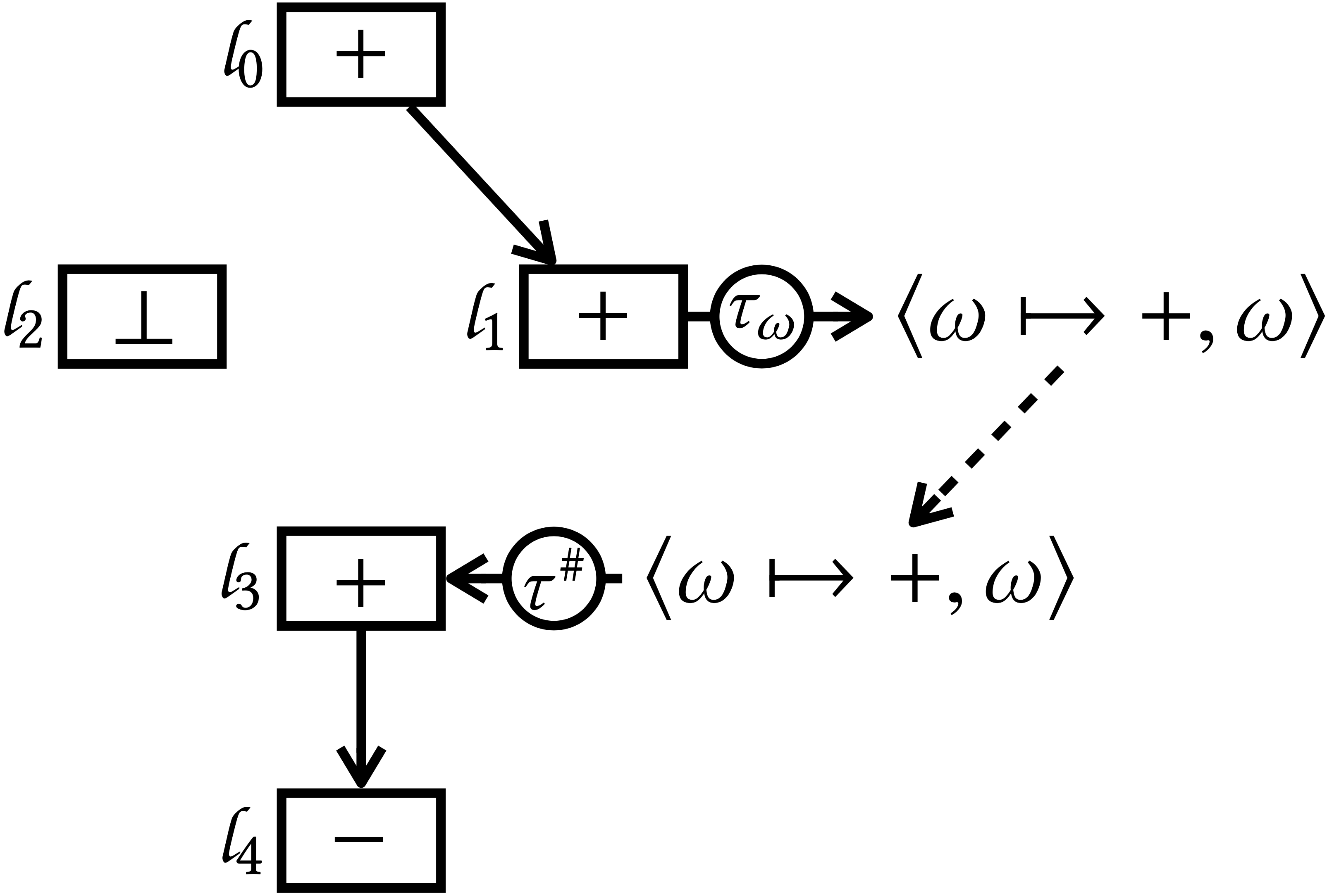}
    \caption{$\varx > 0$}
    \label{fig:ds-example3}
  \end{subfigure}
  \begin{subfigure}[t]{0.28\textwidth}
    \includegraphics[height=3.2cm]{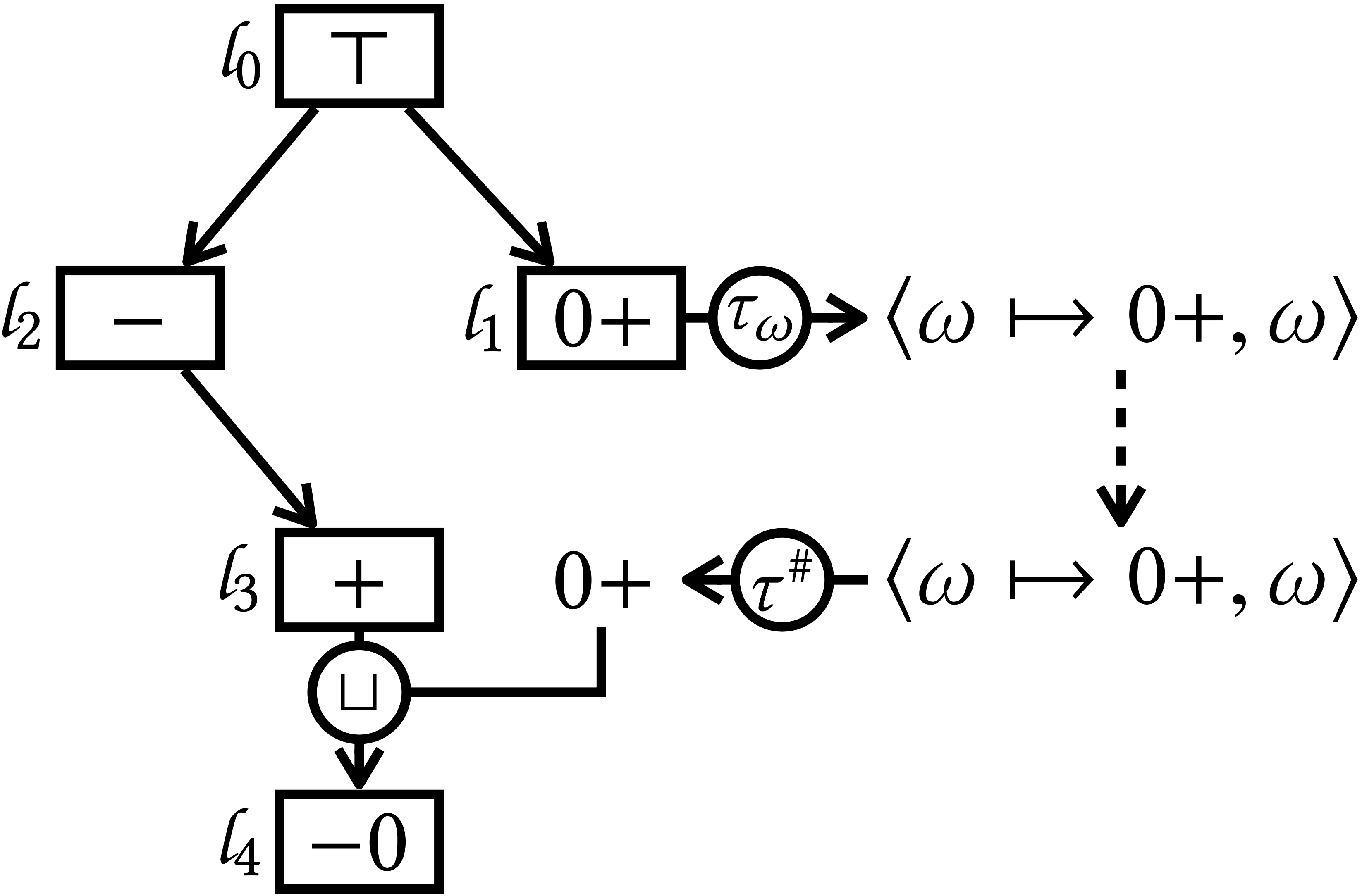}
    \caption{$\varx \in \mathbb{N}$}
    \label{fig:ds-example4}
  \end{subfigure}
  \caption{Abstract interpretation using a combined domain for the running
  example with different initial values for $\varx$.}
  \label{fig:ds-examples}
\end{figure*}

Moreover, to freely convert between different kinds of analysis elements, we define two converters:
\begin{align}
  \asconverter & : (\viewset \times \absdom) \hookrightarrow
    (\absimapset \times \symbstset)\\
  \saconverter & : (\viewset \times
    \absdom) \leftarrow (\absimapset \times \symbstset)
\end{align}
While the converter $\saconverter$ is total, the other one $\asconverter$ is
\textit{partial}. Thus, it is possible to convert an analysis element
$(\view, \abselem)$ in a sensitive abstract domain to another analysis element in
a {\sealed} domain only if the convert $\asconverter$ is defined: $(\view,
\abselem) \in \Dom(\asconverter)$.  In addition, they should convert given
analysis elements without loss of information for all $\aelem \in \aelemset$:
\[
  \asconverter(\aelem) = \aelem' \Rightarrow \left\{
  \begin{array}{l}
    \aelem = \saconverter(\aelem')\\
    \aelemgamma(\aelem) = \aelemgamma(\aelem')\\
  \end{array}
  \right.
\]

Now, we define the \textit{combined one-step execution} $\combstep: \combdom
\rightarrow \combdom$ with two converters $\asconverter$ and $\saconverter$.
It consists of two steps: 1) the \textit{\reformname} step converts
analysis elements if a new {\sealed} execution starts or an
existing one stops, and 2) the \textit{execution} step performs execution of each
analysis element using the abstract one-step execution $\sabsstep$ in the sensitive
abstract domain and the {\sealed} one-step execution $\symbstep$ in the {\sealed} domain.
\begin{definition}[Combined One-Step Execution]
A \textit{combined one-step execution} $\combstep: \combdom \rightarrow
\combdom$ is define as follows:
  \[
    \combstep(\combelem) = (\sabsstep(\sabselem), \symbstep(\symbelem))
  \]
where $(\sabselem, \symbelem) = \reform(\combelem)$.
\end{definition}

From a given combined state $\combelem$, the $\reform$ function makes analysis elements
and converts them if a new {\sealed} execution
begins or an existing {\sealed} execution terminates.
Specifically, for an analysis element $(\view, \abselem)$ in the sensitive abstract domain,
if the converter $\asconverter$ is defined for it, $\reform$ introduces a new {\sealed} execution
by converting the analysis element to its corresponding one $(\absimap, \symbst) =
\asconverter((\view, \abselem))$ in the {\sealed} domain.
On the other hand, for an analysis element $(\absimap, \symbst)$ in the {\sealed} domain,
if it does not have any {\sealed} states to transit to, $\symbst \symbtrans \excst$,
the {\sealed} execution for $(\absimap, \symbst)$ terminates.
It converts the analysis element to its corresponding one $(\view, \abselem) =
\saconverter((\absimap, \symbst))$ in the sensitive abstract domain and
merges the current abstract state stored in $\view$ with $\abselem$.

To formally define the $\reform$ function, we first define a $\areform$ function
for analysis elements using two converters.
\begin{definition}[$\areform$]\label{def:areform}
  The function $\areform: \aelemset \rightarrow \aelemset$ for analysis elements
  is defined as follows:
  \[
    \areform(\aelem) = \left\{
      \begin{array}{ll}
        \asconverter(\aelem)
        & \text{if} \; \aelem = (\view, \abselem) \wedge \aelem \in
        \Dom(\asconverter)\\
        \saconverter(\aelem)
        & \text{if} \; \aelem = (\absimap, \symbst) \wedge \symbst \symbtrans
        \bot\\
        \aelem
        & \text{Otherwise}
      \end{array}
    \right.
  \]
\end{definition}
\begin{definition}[$\reform$]\label{def:reform}
  The \reformname function $\reform: \combdom \rightarrow \combdom$ for combined
  states is defined as follows:
  \[
    \reform((\sabselem, \symbelem)) = \left(
      \lambda \view. \bigjoin \{ \abselem \!\mid\! (\view, \abselem) \in E \},
      E \cap (\absimapset \times \symbstset)
    \right)
  \]
  where
  \[
    E = \dot{\areform}(\{ (\view, \sabselem(\view)) \mid \view \in \viewset \} \cup \symbelem)
  \]
and the dot notation $\dot{f}$ denotes the element-wise extended function of a
function $f$.
\end{definition}

\subsection{Examples}
Now, we show examples of abstract interpretation with a combined domain.
Figure~\ref{fig:ds-examples} depicts the flow of analysis for the running
example in Figure~\ref{fig:running-example} with three different initial sets of
values for the variable $\varx$.  In this example, we use the abstract domain
$\{ -, 0, + \}$ for integers stored in $\varx$ as introduced in
Section~\ref{sec:ai}, and the \textit{flow sensitivity} that utilizes the
labels of states as their views as introduced in Section~\ref{sec:sens}.
For brevity, we use concatenation of abstract values so that
$-0$ denotes the set $\{ -, 0 \}$.

Figure~\ref{fig:ds-examples}(a) presents notations used in each graph. A solid
box denotes an analysis element that is a pair of a label $\lab$ and an abstract
state $\abselem$.  A pair enclosed by angle brackets denotes an analysis
element that is a pair of an abstract instantiation map $\absimap$ and a {\sealed}
state $\symbst$.  In fact, the {\sealed} state part (right) of
each pair in graphs contains only the value of the variable of $\varx$ without
its label.  For brevity, we represent its label by locating it next to
a node with its label.  A solid line is a view transition
$\viewtrans{\lab}{\lab'}$ from a label $\lab$ to another one $\lab'$.  A dotted
line is a {\sealed} transition $\symbtrans$.  Three solid lines with
circled labels denote two converters $\saconverter$, $\asconverter$ and the join
operator $\join$.

Figure~\ref{fig:ds-examples}(b) shows the analysis with the combined domain when
the initial value of $\varx$ is $0$.  First, in the \reformname step,
the converter $\asconverter$ converts the analysis element $(\lab_0, 0)$ to
another analysis element $\langle \varnothing, 0 \rangle$ with the label
$\lab_0$.  It does not introduce any {\sealed} values because
the value represents only a single value.  Until the end of the program, the
{\sealed} execution from $\langle \varnothing, 0 \rangle$ successfully
continues.  Because there is no more possible {\sealed} transition for the
{\sealed} state $\langle \varnothing, 0 \rangle$ with $\lab_4$,
it is converted to $(\lab_4, 0)$ via the converter $\saconverter$.

Instead of a single value, assume that the initial value of $\varx$ is one of
any positive integers.  Figure~\ref{fig:ds-examples}(c) describes the analysis
flow for the case.  The initial abstract value at the label $\lab_0$ is
$+$ and it is impossible to convert it to any {\sealed} values because the
next program statement requires the actual value stored in the variable $\varx$
for the branch condition $\varx \geq 0$.  Thus, it performs view transition
$\viewtrans{\lab_0}{\lab_1}$ from the label $\lab_0$ to another one $\lab_1$ for
the abstract value $+$ and the result is also $+$.  Now, the analysis element
$(\lab_1, +)$ can be converted to $\langle \symb \mapsto +, \symb \rangle$
with the label $\lab_1$.  This {\sealed} execution step terminates in the
label $\lab_3$ because the next statement is $\varx = -\varx$ and the negation
operator requires the actual value of $\varx$.  It is converted to $(\lab_3, +)$ via $\saconverter$,
performs the view transition, and results in $(\lab_4, -)$.

For the last case, we assume that all integers are possible for the initial
value of the variable $\varx$ as described in Figure~\ref{fig:ds-examples}(d).
While it reaches the false branch in the label $\lab_2$ unlike previous cases,
it cannot perform dynamic shortcuts because the statement in the false
branch is $\varx = -\varx$, which requires the actual value of $\varx$.
At the label $\lab_3$, there are two analysis
elements: 1) $(\lab_3, +)$ introduced by the view transition from the label $\lab_2$
with $-$, and 2) $\langle \symb \mapsto 0+, \symb \rangle$ with $\lab_3$
introduced by {\sealed} execution started at $\lab_1$.  Since it
is not possible to perform {\sealed} execution for both elements, the
second one is converted to $(\lab_3, 0+)$ and merged with $+$ at $\lab_3$ via the
join operator $\join$.  Finally, the view transition
$\viewtrans{\lab_3}{\lab_4}$ from $\lab_3$ to $\lab_4$ is performed to the
merged abstract state $0+$ and the result is $-0$.

\subsection{Soundness and Termination}
The converter $\asconverter$ and the {\sealed} transition $\symbtrans$ are
keys to configure the introduction and termination of {\sealed}
execution.  To ensure the \textit{soundness} and \textit{termination} of an
abstract interpretation defined with a combined domain of a sensitive abstract
domain and a {\sealed} domain, the following conditions should hold.

\begin{theorem}[Soundness and Termination]\label{theorem:shortcut}
An abstract interpretation with dynamic shortcuts is \textbf{sound} and
\textbf{terminates} in a finite time if:
  \begin{itemize}
    \item the abstract transfer function $\abstransfer$ is sound,
    \item the sensitive abstract domain $\sabsdom$ has a finite height,
    \item the {\sealed} transition $\symbtrans$ is valid, and
    \item there exists $N < \infty$ such that
      \begin{equation}\label{equ:asc-cond}
        \forall \aelem \in \aelemset. \; \asconverter(\aelem) = (\absimap,
        \symbst) \Rightarrow \symbst
        \symbtrans^k \excst \wedge 1 < k \leq N
      \end{equation}
  \end{itemize}
\end{theorem}

For soundness proof, we should prove
two conditions presented in Section~\ref{sec:ai}:
(\ref{equ:sound-join}) for the join operator $\join$ and
(\ref{equ:sound-step}) for the combined one-step execution.
The core idea of the proof is to use Lemma~\ref{lemma:sound-symbstep} and
Lemma~\ref{lemma:reform} for the {\sealed} one-step execution
$\symbstep$ and the $\reform$ function, respectively.
On the other hand, the core idea of the termination proof is to use the property
that the second and the fourth conditions provide upper bounds of the number of
sensitive abstract states and the number of {\sealed} states, respectively.  We
formally define and prove the property using \textit{time to live (TTL)}
functions of sealed states, $\ttl_i$ for each iteration $i \geq 0$, and prove the
termination using them.  Now, we assume that its all conditions in
Theorem~\ref{theorem:shortcut} are hold and rephrase the \textit{soundness} as
Theorem~\ref{theorem:soundness} and \textit{termination} as
Theorem~\ref{theorem:termination}.

\subsubsection{Soundness}

\begin{theorem}[Soundness]\label{theorem:soundness}
  The abstract interpretation using the combined domain $\combdom$ is
  \textbf{sound} if
  \begin{equation}\label{equ:sound-join}
    \forall \combelem_0, \combelem_1 \in \combdom. \; \combgamma(\combelem_0) \cup
    \combgamma(\combelem_1) \subseteq \combgamma(\combelem_0 \join \combelem_1)
  \end{equation}
  \begin{equation}\label{equ:sound-combstep}
    \forall \combelem \in \combdom. \; \step \circ \combgamma(\combelem) \subseteq
    \combgamma \circ \combstep(\combelem)\\
  \end{equation}
\end{theorem}
\begin{proof}
  First, we prove that the abstract transfer function $\combtransfer: \combdom
  \rightarrow \combdom$ defined as $\combtransfer(\combelem) = \combelem \join
  \combstep(\combelem)$ is sound
  \[
    \begin{array}{rcll}
      \transfer \circ \combgamma(\combelem)
      &=& \combgamma(\combelem) \cup \step \circ \combgamma(\combelem)\\
      &\subseteq& \combgamma(\combelem) \cup \combgamma \circ \combstep(\combelem)
      & (\because \; \text{condition~(\ref{equ:sound-combstep})})\\
      &\subseteq& \combgamma(\combelem \join \combstep(\combelem))
      & (\because \; \text{condition~(\ref{equ:sound-join})})\\
      &=& \combgamma \circ \combtransfer(\combelem)\\
    \end{array}
  \]
  Then, the abstract semantics $\combsem{\prog} = \underset{n \rightarrow
  \infty}{\lim}{(\combtransfer)^n(\icombelem)} $ is also sound because it is
  defined with a sound abstract transfer function $\combtransfer$ using the
  combined one-step execution $\combstep$.
\end{proof}

Now, we should show that two conditions about the soundness of the join
operator (\ref{equ:sound-join}) and the soundness of the combined one-step
execution (\ref{equ:sound-combstep}) in Theorem~\ref{theorem:soundness} hold.

First, we prove the soundness of the join operator (\ref{equ:sound-join}) in
Lemma~\ref{lemma:sound-join}.
\begin{lemma}[Soundness of $\join$]\label{lemma:sound-join}
  \[
    \forall \combelem_0, \combelem_1 \in \combdom. \; \combgamma(\combelem_0) \cup
    \combgamma(\combelem_1) \subseteq \combgamma(\combelem_0 \join \combelem_1)
  \]
\end{lemma}
\begin{proof}
  \[
    \begin{array}{cl}
      \multicolumn{2}{l}{
        \combgamma((\sabselem, \symbelem)) \cup \combgamma((\sabselem', \symbelem'))
      }\\
      =& \sgamma(\sabselem) \cup \symbgamma(\symbelem)
      \cup \sgamma(\sabselem') \cup \symbgamma(\symbelem')\\
      =& (\sgamma(\sabselem)\cup \sgamma(\sabselem'))
      \cup (\symbgamma(\symbelem) \cup \symbgamma(\symbelem'))\\
      \subseteq& \sgamma(\sabselem \join \sabselem')
      \cup (\symbgamma(\symbelem) \cup \symbgamma(\symbelem'))\\
      & \multicolumn{1}{r}{(\because \; \sabsdom \; \text{is sound})}\\
      =& \sgamma(\sabselem \join \sabselem')
      \cup \symbgamma(\symbelem \cup \symbelem')\\
      =& \combgamma((\sabselem \join \sabselem', \symbelem \cup \symbelem'))\\
      =& \combgamma((\sabselem, \symbelem) \join (\sabselem', \symbelem'))\\
    \end{array}
  \]
\end{proof}

For the condition (\ref{equ:sound-combstep}), we first prove two properties of the
$\reform$ function in Lemma~\ref{lemma:reform}.  Using the properties, we prove
the soundness of the sealed one-step execution in
Lemma~\ref{lemma:sound-symbstep}.  Finally, we prove the soundness of the
combined one-step execution (\ref{equ:sound-combstep}) in
Lemma~\ref{lemma:sound-combstep}.

\begin{lemma}[Properties of $\reform$]\label{lemma:reform}
  For a given combined state $\combelem \in \combdom$, the $\reform$ function
  satisfies the following two properties:
  \begin{itemize}
    \item $\combgamma(\combelem) \subseteq \combgamma \circ \reform(\combelem)$
    \item $\forall (\absimap, \symbst) \in \symbelem. \; \exists \symbst' \in
      \symbstset.  \; \text{s.t.} \; \symbst \symbtrans \symbst'$
  \end{itemize}
  where $(\sabselem, \symbelem) = \reform(\combelem)$
\end{lemma}
\begin{proof}
  \[
    \fbox{$\combgamma(\combelem) \subseteq \combgamma \circ \reform(\combelem)$}
  \]
  \[
    \begin{array}{cl}
      \multicolumn{2}{l}{\combgamma((\sabselem, \symbelem))}\\
      =& \sgamma(\sabselem) \cup \symbgamma(\symbelem)\\

      =& \left( \underset{\view \in \viewset}{\bigcup} {\viewmap(\view) \cap
      \gamma \circ \sabselem(\view)} \right) \cup \left( \underset{(\absimap,
      \symbst) \in \symbelem}{\bigcup} \instant{\symbst}{\absimap} \right) \\

      =& \left( \underset{\view \in \viewset}{\bigcup} \aelemgamma((\view,
      \sabselem(\view))) \right) \cup \left( \underset{(\absimap, \symbst) \in
      \symbelem}{\bigcup} \aelemgamma((\absimap, \symbst)) \right) \\

      =& \dot\aelemgamma(\{ (\view, \sabselem(\view)) \mid \view \in \viewset \}
      \cup \symbelem)\\

      =& \dot\aelemgamma(\dot{\areform}(\{ (\view, \sabselem(\view)) \mid \view
      \in \viewset \} \cup \symbelem))\\

       & \multicolumn{1}{r}{\because \; (\text{Trivially,} \; \forall \aelem \in
       \aelemset.  \; \aelemgamma(\aelem) = \aelemgamma \circ \areform(\aelem))}\\

      =& \dot\aelemgamma(E)\\
       & \multicolumn{1}{r}{\because \; (\text{See the definition of $E$ in
       Definition~\ref{def:reform}})}\\
    \end{array}
  \]
  \[
    \begin{array}{cl}
      =& \left( \underset{(\view, \abselem) \in E}{\bigcup} \aelemgamma((\view,
      \abselem)) \right) \cup \left( \underset{(\absimap, \symbst) \in
      E}{\bigcup} \aelemgamma((\absimap, \symbst)) \right)\\

      =& \left( \underset{(\view, \abselem) \in E}{\bigcup} \viewmap(\view) \cap
      \gamma(\abselem) \right) \cup \left( \underset{(\absimap, \symbst) \in
      E}{\bigcup} \instant{\symbst}{\absimap} \right)\\

      =& \left( \underset{\view \in \viewset}{\bigcup} { \underset{(\view,
      \abselem) \in E}{\bigcup} \viewmap(\view) \cap \gamma(\abselem) } \right)
      \cup \left( \underset{(\absimap, \symbst) \in E}{\bigcup}
      \instant{\symbst}{\absimap} \right)\\

      =& \left( \underset{\view \in \viewset}{\bigcup} {\viewmap(\view) \cap
        \left(\underset{(\view, \abselem) \in
      E}{\bigcup}{\gamma(\abselem)}\right)} \right) \cup \left(
      \underset{(\absimap, \symbst) \in E}{\bigcup} \instant{\symbst}{\absimap}
      \right)\\

      \subseteq& \left( \underset{\view \in \viewset}{\bigcup} {\viewmap(\view)
        \cap \gamma\left( \underset{(\view, \abselem) \in E}{\bigjoin}\abselem
      \right)} \right) \cup \left( \underset{(\absimap, \symbst) \in E}{\bigcup}
      \instant{\symbst}{\absimap} \right)\\

      =& \sgamma\left( \lambda \view. \underset{(\view, \abselem) \in
      E}{\bigjoin}\abselem \right) \cup \symbgamma(E \cap (\absimapset \times
      \symbstset))\\

      =& \combgamma\left( \lambda \view. \underset{(\view, \abselem) \in
      E}{\bigjoin}\abselem, E \cap (\absimapset \times \symbstset) \right)\\

      =& \combgamma \circ \reform((\sabselem, \symbelem))
    \end{array}
  \]
  \[\]
  \[
    \fbox{$\forall (\absimap, \symbst) \in \symbelem. \; \exists \symbst' \in
    \symbstset.  \; \text{s.t.} \; \symbst \symbtrans \symbst'$}
  \]

  For a given $(\absimap, \symbst) \in \symbelem$, there exists an analysis
  element $\aelem \in \aelemset$ such that $\areform(\aelem) = (\absimap,
  \symbst)$.  According to the definition of $\areform$ in
  Definition~\ref{def:areform}, there are two possible cases: $\aelem =
  (\absimap, \symbst) \wedge \exists \symbst' \in \symbstset. \; \text{s.t} \;
  \symbst \symbtrans \symbst'$ or $\aelem = (\view, \abselem) \wedge \aelem \in
  \Dom(\asconverter)$. We separately consider those two cases:
  \begin{itemize}
    \item $\aelem = (\absimap, \symbst) \wedge \exists \symbst' \in \symbstset.
      \; \text{s.t} \; \symbst \symbtrans \symbst'$\\
        By definition, $\exists \symbst' \in \symbstset.  \; \text{s.t} \;
        \symbst \symbtrans \symbst'$
    \item $\aelem = (\view, \abselem) \wedge \aelem \in \Dom(\asconverter)$\\
      By the condition~\ref{equ:asc-cond} in the Theorem~\ref{theorem:shortcut},\\
      $\exists k > 1. \symbst \symbtrans^k \excst$.  Thus, $\exists \symbst' \in
      \symbstset.  \; \text{s.t} \; \symbst \symbtrans \symbst'$
  \end{itemize}
\end{proof}

\begin{lemma}[Soundness of $\symbstep$]\label{lemma:sound-symbstep}
  The sealed one-step execution $\symbstep$ is sound:
  \[
    \step \circ \symbgamma(\symbelem) \subseteq
    \symbgamma \circ \symbstep(\symbelem)\\
  \]
\end{lemma}
\begin{proof}
  \[
    \begin{array}{cl}
      \multicolumn{2}{l}{\step \circ \symbgamma(\symbelem)}\\
      =& \step(\bigcup \{ \instant{\symbst}{\absimap} \mid (\absimap, \symbst) \in
      \symbelem\})\\
      =& \{ \st' \mid (\absimap, \symbst) \in \symbelem \wedge \st \in
      \instant{\symbst}{\absimap} \wedge \st \trans \st'\})\\
      =& \{ \st' \mid (\absimap, \symbst) \in \symbelem \wedge \imap \in
      \imapgamma(\absimap) \wedge \instant{\symbst}{\imap} \trans \st'\})\\
      =& \{ \st' \mid (\absimap, \symbst) \in \symbelem \wedge \imap \in
      \imapgamma(\absimap) \wedge \instant{\symbst}{\imap} \trans \st'\\
       & \phantom{\{ \st' \mid (\absimap, \symbst) \in \symbelem \wedge \imap \in
      \imapgamma(\absimap)} \wedge \symbst \symbtrans \symbst' \})\\
       & \multicolumn{1}{r}{(\because \; \text{Second property in
       Lemma~ref{lemma:reform}})}\\
    \end{array}
  \]
  \[
    \begin{array}{cl}
      =& \{ \instant{\symbst'}{\imap} \mid (\absimap, \symbst)
      \in \symbelem \wedge \imap \in \imapgamma(\absimap) \wedge \symbst
      \symbtrans \symbst'\}\\
      & \multicolumn{1}{r}{(\because \; \text{Validity of} \; \symbtrans)}\\
      =& \bigcup \{ \instant{\symbst'}{\absimap} \mid (\absimap, \symbst)
      \in \symbelem \wedge \symbst \symbtrans \symbst' \}\\
      =& \symbgamma(\{ (\absimap, \symbst') \mid (\absimap, \symbst)
      \in \symbelem \wedge \symbst \symbtrans \symbst' \})\\
      =& \symbgamma \circ \symbstep(\symbelem)\\
    \end{array}
  \]
\end{proof}

\begin{lemma}[Soundness of $\combstep$]\label{lemma:sound-combstep}
  The combined one-step execution $\combstep$ is sound:
  \[
    \forall \combelem \in \combdom. \; \step \circ \combgamma(\combelem) \subseteq
    \combgamma \circ \combstep(\combelem)\\
  \]
\end{lemma}
\begin{proof}
  \[
    \begin{array}{cl}
      \multicolumn{2}{l}{
        \step \circ \combgamma(\combelem)
      }\\
      \subseteq& \step \circ \combgamma((\sabselem, \symbelem))\\
       & \multicolumn{1}{r}{(\because \; \text{First property in
       Lemma~\ref{lemma:reform}}}\\
       & \multicolumn{1}{r}{\text{where} \; (\sabselem, \symbelem)
       = \reform(\combelem).)}\\

      =& \step(\sgamma(\sabselem) \cup \symbgamma(\symbelem))\\
      =& \step(\sgamma(\sabselem)) \cup \step(\symbgamma(\symbelem))\\
      \subseteq& \sgamma \circ \sabsstep(\sabselem) \cup \step(\symbgamma(\symbelem))\\
      & \multicolumn{1}{r}{(\because \; \sabsstep \; \text{is sound.})}\\
      \subseteq& \sgamma \circ \sabsstep(\sabselem) \cup \symbgamma \circ
      \symbstep((\symbelem))\\
               & \multicolumn{1}{r}{(\because \; \text{and
               Lemma~\ref{lemma:sound-symbstep}})}\\
      =& \combgamma((\sabsstep(\sabselem), \symbstep(\symbelem)))\\
      =& \combgamma \circ \combstep(\combelem)\\
    \end{array}
  \]
\end{proof}

\subsection{Termination}

Before proving the termination of the abstract interpretation using the combined
domain $\combdom$, we define several notations. The initial abstract state
$\icombelem = (\isabselem, \varnothing)$ is pair of the initial abstract state of
the sensitive abstract domain $\sabsdom$ and an empty set. For each iteration $i
\geq 0$, we define the $i$-th result of abstract interpretation
$\combtransfer^i(\icombelem) = \combelem^i = (\sabselem^i, \symbelem^i)$ and the
\textit{difference set} $\diffset_i = \symbelem^{i+1} \setminus \symbelem^i$.
For simplicity, we define $\diffset_i$ as $\varnothing$ for $i < 0$.  Moreover, we
define a lifted version of sealed relation $\liftsymbtrans \subseteq
(\absimapset \times \symbstset) \times (\absimapset \times \symbstset)$ as
follows:
\[
  (\absimap, \symbst) \liftsymbtrans (\absimap, \symbst') \Leftrightarrow
  \symbst \symbtrans \symbst'
\]
Using the lifted relation, we define the \textit{time to live (TTL)} function of
sealed states $\ttl_i: \diffset_i \rightarrow \numset$ for each iteration $i
\geq 0$ as follows:
\begin{definition}[TTL Function]
  \[
    \begin{array}{c}
      \ttl_i(\symbaelem) = \left \{
      \begin{array}{l}
        N - 1 \;\; ( \text{if} \; D = \varnothing)\\
        \text{min}(\dot\ttl_{i-1}(D)) - 1
        \;\; ( \text{otherwise})
      \end{array}
      \right. \\
      \\
      \text{where} \; D =
      \{\symbaelem' \in \diffset_{i-1} \mid \symbaelem' \liftsymbtrans \symbaelem\}
    \end{array}
  \]
\end{definition}

Based on the notations, we formally prove the termination property as follows:
\begin{theorem}[Termination]\label{theorem:termination}
  The abstract interpretation using the combined domain $\combdom$
  \textbf{terminates} in a finite time if
  \begin{equation}\label{equ:termination-sai}
    \exists n. \; \forall m \geq n. \; \sabselem^m = \sabselem^n
  \end{equation}
  \begin{equation}\label{equ:bounded-ttl}
    \forall i \geq 0. \; \forall \symbaelem \in \diffset_i. \;
    0 < \ttl_i(\symbaelem) < N
  \end{equation}
  \begin{equation}\label{equ:dec-ttl}
    \begin{array}{c}
      \forall i > 0. \; \sabselem^{i-1} = \sabselem^i \Rightarrow\\
      \sup(\dot \ttl_i(\diffset_i)) \leq \sup(\dot \ttl_{i-1}(\diffset_{i-1})) - 1
    \end{array}
  \end{equation}
\end{theorem}

\begin{proof}
  By the condition (\ref{equ:termination-sai}), there exists $n \in \numset$
  such that $\sabselem^m = \sabselem^n$ for all $m \geq n$.  By the condition
  (\ref{equ:bounded-ttl}), the TTL of each sealed state in $\diffset_n$ is
  bounded by $N$:
  \[
    \sup(\dot \ttl_n(\diffset_n)) < N
  \].
  Then, the upper bound of TTL for sealed states in each difference set after
  the $n-$th iteration is decreased by the condition (\ref{equ:dec-ttl}):
  \[
    \forall i > 0. \; \sup(\dot \ttl_{n+i}(\diffset_{n+i})) \leq \sup(\dot
    \ttl_{n+i-1}(\diffset_{n+i-1})) - 1
  \].
  which implies that
  \[
    \sup(\dot \ttl_{n+i}(\diffset_{n+i})) \leq \sup(\dot \ttl_n(\diffset_n)) - i < N - i
  \]
  Therefore, for $j \geq N$,
  \[
    \sup(\dot \ttl_{n+j}(\diffset_{n+j})) < N - j \leq 0
  \]
  Notice that again by the condition (\ref{equ:bounded-ttl}),
  \[
    inf(\dot \ttl_{n+j}(\diffset_{n+j})) > 0
  \]
  meaning that
  \[
    inf(\dot \ttl_{n+j}(\diffset_{n+j})) > \sup(\dot \ttl_{n+j}(\diffset_{n+j}))
  \]
  which implies $\diffset_{n+j} = \varnothing$ and $\symbelem^{n+j+1} =
  \symbelem^{n+j}$.
  Therefore, for all $m \geq n + N$,
  \[
    \sabselem^m = \sabselem^{n+N} \wedge \symbelem^m = \symbelem^{n+N}
  \]
  and
  \[
    \combelem^m = \combelem^{n+N}
  \]
  which means the abstract interpretation using the combined domain
  $\combdom$ terminates in $n+N$ iterations.
\end{proof}

Now, we should show that three conditions about the termination of the sensitive
abstract interpretation (\ref{equ:termination-sai}), the bound of TTL for
sealed states in difference sets (\ref{equ:bounded-ttl}), and the decrease of
their upper bounds (\ref{equ:dec-ttl}) in Theorem~\ref{theorem:termination}
hold.

First, we prove the termination of the sensitive abstract interpretation
(\ref{equ:termination-sai}) in Lemma~\ref{lemma:termination-sai}.

\begin{lemma}[Termination of Sensitive Abstract Interpretation]\label{lemma:sabs-term}
\label{lemma:termination-sai}
  \[
    \exists n. \; \forall m \geq n. \;
    \sabselem^m = \sabselem^n
  \]
\end{lemma}

\begin{proof}
Note that for all $\sabselem, \sabselem' \in \sabsdom$ that satisfies
$\combtransfer((\sabselem, \_)) = (\sabselem', \_)$,
\[
  \begin{array}{rcl}
  \combtransfer((\sabselem, \_))
  &=& (\sabselem, \_) \join \combstep((\sabselem, \_))\\
  &=& (\sabselem, \_) \join (\_, \_) = (\sabselem \join \_, \_)\\
  &=& (\sabselem', \_)
  \end{array}
\]
which implies $\sabselem \order \sabselem'$.  Since
$\combtransfer((\sabselem^i, \_)) = (\sabselem^{i+1}, \_), \sabselem^i \order
\sabselem^{i+1}$ holds for all $i \geq 0$.  Then, $\sabselem^0 \order \sabselem^1
\order \sabselem^2 \cdots$ is an ascending chain.  Since the height of the
sensitive abstract domain $\sabsdom$ is finite, the ascending chain condition is
also hold. Therefore, there exists n such that for all $m \geq n, \sabselem^m =
\sabselem^n$.
\end{proof}

Then, we prove two remaining conditions (\ref{equ:bounded-ttl}) and
(\ref{equ:dec-ttl}).  We first prove two properties of difference sets in
Lemma~\ref{lemma:diffset_prop} and Corollary~\ref{corollary:only-from-diffset},
and a property of TTL in Lemma~\ref{lemma:prop-ttl}.  Using them, we prove the
bound of TTL for sealed states in difference sets (\ref{equ:bounded-ttl}) in
Corollary~\ref{corollary:bounded-ttl} and the decrease of their upper bounds
(\ref{equ:dec-ttl}) in Lemma~\ref{lemma:dec-ttl}.

\begin{lemma}\label{lemma:diffset_prop}
  \[
    \begin{array}{c}
      \forall i \geq 0. \; \forall \symbaelem \in \diffset_i. \\
      \exists \view. \; \asconverter((\view,\sabselem^i(\view))) \liftsymbtrans \symbaelem 
      \lor \exists \symbaelem' \in \diffset_{i-1} . \; \symbaelem' \liftsymbtrans \symbaelem
    \end{array}
  \]
\end{lemma}
\begin{proof}
  Let $i \in \numset$ and $\symbaelem \in \diffset_i = \symbelem^{i+1} \setminus \symbelem^i$ given.
  By definition,
  \[
    \symbelem^{i+1} = \symbelem^i \cup \symbstep({\symbelem^i}')
  \]
  where
  \[
    (\_, {\symbelem^i}') = \reform(\sabselem^i, \symbelem^i)
  \]
  Note that $\symbaelem \in \symbstep({\symbelem^i}')$,
  and by definition of $\symbstep$, there exists some
  $\symbaelem' \in {\symbelem^i}'$ that satisfies $\symbaelem' \liftsymbtrans \symbaelem$.
  Now, by definition of $\reform$,
  \[
    {\symbelem^i}' =
    \dot{\areform}(\{ (\view, \sabselem^i(\view)) \mid \view \in \viewset \} \cup \symbelem^i)
    \cap (\absimapset \times \symbstset)
  \]
  This means there exists
  $\aelem \in \{ (\view, \sabselem^i(\view)) \mid \view \in \viewset \} \cup \symbelem^i$
  that satisfies $\areform(\aelem) = \symbaelem'$. We have two possible cases for $\aelem$.
  \begin{itemize}
 
  \item $\aelem \in \{ (\view, \sabselem^i(\view)) \mid \view \in \viewset \}$

  In this case, $\areform(\aelem) = \asconverter(\aelem) = \symbaelem'$
  and the left condition for conclusion is satisfied.
  
  \item $\aelem \in \symbelem^i$

  In this case, $\areform(\aelem) = \aelem = \symbaelem'$.
  Now, let's assume that $\aelem \in \symbelem^{i-1}$.
  In that case, $\aelem$ would be preserved after reform step, that is,
  $\aelem \in {\symbelem^{i-1}}'$. Then, by definition of $\symbstep$,
  $\symbaelem \in \symbstep({\symbelem^{i-1}}') \subseteq \symbelem^i$
  which contradicts to the fact that $\symbaelem \in \diffset_i$.
  Therefore, $\aelem \notin \symbelem^{i-1}$, that is,
  $\aelem \in \symbelem^i \setminus \symbelem^{i-1} = \diffset_i$,
  and the right condition for conclusion is satisfied.
  \end{itemize}
\end{proof}

\begin{corollary}\label{corollary:only-from-diffset}
  \[
    \begin{array}{c}
      \forall i > 0. \; \sabselem^{i-1} = \sabselem^i \Rightarrow
      \forall \symbaelem \in \diffset_i. \\
      \exists \symbaelem' \in \diffset_{i-1} . \; \symbaelem' \liftsymbtrans \symbaelem
    \end{array}
  \]
\end{corollary}
\begin{proof}
  The proof goes same as the previous lemma, until the point where we divide
  the case for $\aelem$. Let's assume that the first case holds, that is,
  \[
    \aelem \in \{ (\view, \sabselem^i(\view)) \mid \view \in \viewset \}
  \]
  Since $\sabselem^{i-1} = \sabselem^i$,
  \[
    \aelem \in \{ (\view, \sabselem^{i-1}(\view)) \mid \view \in \viewset \}
  \]
  In that case, $\aelem$ would be transformed after reform step, that is,
  $\asconverter(\aelem) = \symbaelem' \in {\symbelem^{i-1}}'$.
  Then, by definition of $\symbstep$,
  $\symbaelem \in \symbstep({\symbelem^{i-1}}') \subseteq \symbelem^i$
  which contradicts to the fact that $\symbaelem \in \diffset_i$.
  Therefore, only second case holds and the right conclusion in previous lemma is satisfied.
\end{proof}

\begin{lemma}[Property of TTL]\label{lemma:prop-ttl}
  \[
    \begin{array}{c}
      \forall i \geq 0. \; \forall \symbaelem \in \diffset_i.
      \ttl_i(\symbaelem) = k \Rightarrow \\
      k < N \wedge
      \exists (\view, \abselem). \; (\asconverter((\view,\abselem))
      \liftsymbtrans^{(N - k)} \symbaelem) \\
    \end{array}
  \]
\end{lemma}
\begin{proof}
  We prove by induction on $i$.
  Let $\symbaelem \in \diffset_i$.
  \begin{itemize}
  \item If $i = 0$, $\ttl_0(\symbaelem) = N - 1 < N$
  and since only left conclusion of lemma~\ref{lemma:diffset_prop} can hold,
  there exists view $\view$ s.t.
  $\asconverter(\view, \sabselem^0(\view)) \liftsymbtrans^1 \symbaelem$.
  \item If $i > 0$, we have two cases for $D = 
    \{\symbaelem' \in \diffset_{i-1} \mid \symbaelem' \liftsymbtrans \symbaelem\}$.
  If $D = \varnothing$, the argument is similar as $i = 0$ case.
  Otherwise, let $\symbaelem' = \underset{\x \in D}{argmin}{\ttl_{i-1}(x)}$.
  
  By induction hypothesis, we have
  \[
    k' = \ttl_{i-1}(\symbaelem') < N
  \]
  and there exists $(\view, \abselem)$ such that
  \[
    \symbaelem'' = \asconverter((\view,\abselem)) \liftsymbtrans^{(N - k')} \symbaelem'.
  \]
  By definition of $\ttl_i$,
  $\ttl_i(\symbaelem) = \ttl_i(\symbaelem') - 1$, and $k = k' - 1$.
  Then,
  \[
    k = k' - 1 < N - 1 < N
  \]
  and
  $\symbaelem'' \liftsymbtrans^{(N - k - 1)} \symbaelem'$ with
  $\symbaelem' \liftsymbtrans \symbaelem$ implies that
  \[
    \symbaelem'' \liftsymbtrans^{(N - k)} \symbaelem.
  \]
  \end{itemize}
\end{proof}
\begin{corollary}\label{corollary:bounded-ttl}
  \[
    \forall i \geq 0. \; \forall \symbaelem \in \diffset_i. \;
    0 < \ttl_i(\symbaelem) < N
  \]
\end{corollary}
\begin{proof}
We already proved $k = \ttl_i(\symbaelem) < N$.
Now, let's assume that $k \leq 0$.
By previous lemma, there exists $(\view, \abselem)$ such that
\[
  (\asconverter((\view,\abselem)) \liftsymbtrans^{(N - k)} \symbaelem)
\]
Since $N - k \geq N$, this implies that there exists $\symbaelem'$ such that
\[
  (\asconverter((\view,\abselem)) \liftsymbtrans^N \symbaelem')
\]
However, this contradicts to the condition~(\ref{equ:asc-cond}) of $\asconverter$ that says
if $(\view,\abselem)$ is in domain of $\asconverter$,
the number of possible $\symbtrans$ from state of $\asconverter((\view,\abselem))$
is at most $N - 1$.
Therefore, $k > 0$.
\end{proof}

\begin{lemma}\label{lemma:dec-ttl}
  \[
    \begin{array}{c}
      \forall i > 0. \; \sabselem^{i-1} = \sabselem^i \Rightarrow \\
      \sup(\dot \ttl_i(\diffset_i)) \leq \sup(\dot \ttl_{i-1}(\diffset_{i-1})) - 1
    \end{array}
  \]
\end{lemma}
\begin{proof}
  Let $\symbaelem \in \diffset_i$.
  By Corollary~\ref{corollary:only-from-diffset}, the set
  \[
    D = \{\symbaelem' \in \diffset_{i-1} \mid \symbaelem' \liftsymbtrans \symbaelem\}
  \]
  is non-empty, and for some $\symbaelem' \in \diffset_{i-1}$,
  \[
    \ttl_i(\symbaelem) = \ttl_{i-1}(\symbaelem') - 1 \leq \sup(\dot \ttl_{i-1}(\diffset_{i-1})) - 1
  \]
  Since it holds for every $\symbaelem \in \diffset_i$,
  \[
    \sup(\dot \ttl_i(\diffset_i)) \leq \sup(\dot \ttl_{i-1}(\diffset_{i-1})) - 1
  \]
\end{proof}

%% file: javascript.tex
\section{Dynamic Shortcuts for JavaScript}\label{sec:javascript}
In this section, we introduce the core language of JavaScript that supports
first-class functions, open objects, and first-class property names, and define
{\sealed} execution of the core language for dynamic shortcuts.

\subsection{Core Language of JavaScript}
\begin{figure*}[t]
  \centering

  \fbox{$\st \trans \st$}
  \begin{mathpar}
    \inferrule*[width=0.48\textwidth]
    {
      \prog(\lab) = \refer = \expr\\
      \referrule{\st}{\refer}{\loc}\\
      \exprrule{\st}{\expr}{\val}\\
    }
    {
      \st = (\lab, \mem, \ctxt, \addr)
      \trans
      (\labnext(\lab), \mem[\loc \mapsto \val], \ctxt, \addr)
    }

    \inferrule*[width=0.48\textwidth]
    {
      \prog(\lab) = \refer = \kwobj\\
      \referrule{\st}{\refer}{\loc}\\
      \addr' = \text{(a fresh object address)}
    }
    {
      \st = (\lab, \mem, \ctxt, \addr)
      \trans
      (\labnext(\lab), \mem[\loc \mapsto \addr'], \ctxt, \addr)
    }

    \inferrule*[width=0.48\textwidth]
    {
      \prog(\lab) = \refer = \expr_f ( \expr_a )\\
      \referrule{\st}{\refer}{\loc}\\
      \exprrule{\st}{\expr_f}{\fval{x}{\lab_b}}\\
      \exprrule{\st}{\expr_a}{\val_a}\\
      \addr' = \text{(a fresh environment address)}
    }
    {
      \st = (\lab, \mem, \ctxt, \addr)
      \trans
      (\lab_b, \mem[(\addr', x) \mapsto \val_a], \ctxt[\addr' \mapsto (\addr,
      \labnext(\lab), \loc)], \addr')
    }

    \inferrule*[width=0.48\textwidth]
    {
      \prog(\lab) = \kwret \; \expr\\
      \exprrule{\st}{\expr}{\val}\\
      \ctxt(\addr) = (\addr', \lab', \loc)
    }
    {
      \st = (\lab, \mem, \ctxt, \addr)
      \trans
      (\lab', \mem[\loc \mapsto \val], \ctxt, \addr')
    }

    \inferrule*[width=0.48\textwidth]
    {
      \prog(\lab) = \kwif \; \expr \; \lab'\\
      \exprrule{\st}{\expr}{\kwtrue}\\
    }
    {
      \st = (\lab, \mem, \ctxt, \addr)
      \trans
      (\lab', \mem, \ctxt, \addr)
    }

    \inferrule*[width=0.48\textwidth]
    {
      \prog(\lab) = \kwif \; \expr \; \lab'\\
      \exprrule{\st}{\expr}{\kwfalse}\\
    }
    {
      \st = (\lab, \mem, \ctxt, \addr)
      \trans
      (\labnext(\lab), \mem, \ctxt, \addr)
    }
  \end{mathpar}

  \fbox{$\referrule{\st}{\refer}{\loc}$}
  \begin{mathpar}
    \inferrule*[width=0.48\textwidth]
    {}
    {
      \referrule{\st = (\lab, \mem, \ctxt, \addr)}{x}{(\addr, x)}\\
    }

    \inferrule*[width=0.48\textwidth]
    {
      \exprrule{\st}{\expr_0}{\addr_0}\\
      \exprrule{\st}{\expr_1}{\val_1}\\
      \val_1 \in \strset\\
    }
    {
      \referrule{\st = (\lab, \mem, \ctxt, \addr)}{\expr_0 [ \expr_1
      ]}{(\addr_0, \val_1)}
    }
  \end{mathpar}

  \fbox{$\exprrule{\st}{\expr}{\val}$}
  \begin{mathpar}
    \inferrule*[width=0.48\textwidth]
    {}
    {
      \exprrule{\st = (\lab, \mem, \ctxt, \addr)}{\pval}{\pval}
    }

    \inferrule*[width=0.48\textwidth]
    {}
    {
      \exprrule{\st = (\lab, \mem, \ctxt,
      \addr)}{\fval{x}{\lab'}}{\fval{x}{\lab'}}
    }

    \inferrule*[width=0.48\textwidth]
    {
      \referrule{\st}{\refer}{\loc}\\
      \loc \in \Dom(\mem)
    }
    {
      \exprrule{\st = (\lab, \mem, \ctxt, \addr)}{\refer}{\mem(\loc)}
    }

    \inferrule*[width=0.48\textwidth]
    {
      \exprrule{\st}{\expr_1}{\val_1}\\
      \cdots\\
      \exprrule{\st}{\expr_n}{\val_n}
    }
    {
      \exprrule{\st = (\lab, \mem, \ctxt, \addr)}
      {\op(\expr_1, \cdots, \expr_n)}{\op(\val_1, \cdots, \val_n)}
    }
  \end{mathpar}
  \vspace*{-1.5em}
  \caption{The transition relation for the core language of JavaScript}
  \vspace*{-.5em}
  \label{fig:core-trans-rel}
\end{figure*}

\begin{figure*}[t]
  \centering

  \fbox{$\viewtrans{\view}{\view'}: \absdom \rightarrow \absdom$}
  \begin{mathpar}
    \inferrule*[width=0.48\textwidth]
    {
      \abselem = (\absmem, \absctxt, \absaddr, \abscount)\\
      \prog(\lab) = \refer = \expr\\
      \referabssem{\refer}(\abselem) = L\\
      \exprabssem{\expr}(\abselem) = \absval
    }
    {
      \viewtrans{\lab}{\labnext(\lab)}(\abselem) =
      (\absmem[L \mapstos \absval], \absctxt, \absaddr, \abscount)
    }

    \inferrule*[width=0.48\textwidth]
    {
      \abselem = (\absmem, \absctxt, \absaddr, \abscount)\\
      \prog(\lab) = \refer = \kwobj\\
      \referabssem{\refer}(\abselem) = L\\
      \oabsaddr = \lab
    }
    {
      \viewtrans{\lab}{\labnext(\lab)}(\abselem) = (\absmem[L \mapstos \{
      \oabsaddr \}], \absctxt, \absaddr, \inc(\abscount, \oabsaddr))
    }

    \inferrule*[width=0.48\textwidth]
    {
      \abselem = (\absmem, \absctxt, \absaddr, \abscount)\\
      \prog(\lab) = \refer = \expr_f ( \expr_a )\\
      \referabssem{\refer}(\abselem) = L\\
      \fval{x}{\lab_b} \in \exprabssem{\expr_f}(\abselem)\\
      \exprabssem{\expr_a}(\abselem) = \absval_a\\
      \eabsaddr = \lab_b\\
      \eabsctxt = \absctxt[\eabsaddr \mapsto \absctxt(\eabsaddr) \cup \{
      (\absaddr, \labnext(\lab), L) \} ]
    }
    {
      \viewtrans{\lab}{\lab_b}(\abselem) = (\absmem[(\eabsaddr, x) \mapsto
      \absval_a], \eabsctxt, \eabsaddr, \inc(\abscount, \eabsaddr))
    }

    \inferrule*[width=0.48\textwidth]
    {
      \abselem = (\absmem, \absctxt, \absaddr, \abscount)\\
      \prog(\lab) = \kwret \; \expr\\
      \exprabssem{\expr}(\abselem) = \absval\\
      (\rabsaddr, \lab', L) \in \absctxt(\absaddr)
    }
    {
      \viewtrans{\lab}{\lab'}(\abselem) = (\absmem[L \mapstos \absval],
      \absctxt, \rabsaddr, \abscount)
    }

    \inferrule*[width=0.48\textwidth]
    {
      \prog(\lab) = \kwif \; \expr \; \lab'\\
      \kwtrue \in \exprabssem{\expr}(\abselem)
    }
    {
      \viewtrans{\lab}{\lab'}(\abselem) = \abselem
    }

    \inferrule*[width=0.48\textwidth]
    {
      \prog(\lab) = \kwif \; \expr \; \lab'\\
      \kwfalse \in \exprabssem{\expr}(\abselem)
    }
    {
      \viewtrans{\lab}{\labnext(\lab)}(\abselem) = \abselem
    }
  \end{mathpar}

  \fbox{$\referabssem{\refer}: \absdom \rightarrow \powerset{\abslocset}$}
  \begin{mathpar}
    \inferrule*[width=0.48\textwidth]
    {
      \abselem = (\absmem, \absctxt, \absaddr, \abscount)
    }
    {
      \referabssem{x}(\abselem) = \{ (\absaddr, x) \}
    }

    \inferrule*[width=0.48\textwidth]
    {
      \abselem = (\absmem, \absctxt, \absaddr, \abscount)\\
      A = \exprabssem{\expr_0}(\abselem) \cap \absaddrset\\
      S = \exprabssem{\expr_1}(\abselem) \cap \strset\\
    }
    {
      \referabssem{\expr_0 [ \expr_1 ]}(\abselem) = A \times S
    }
  \end{mathpar}

  \fbox{$\exprabssem{\expr}: \absdom \rightarrow \absvalset$}
  \begin{mathpar}
    \inferrule*[width=0.48\textwidth]
    {}
    {
      \exprabssem{\pval}(\abselem) = \{ \pval \}
    }

    {}
    {
      \exprabssem{\fval{x}{\lab}}(\abselem) = \{ \fval{x}{\lab} \}
    }

    \inferrule*[width=0.48\textwidth]
    {
      \abselem = (\absmem, \absctxt, \absaddr, \abscount)\\
      \absval = \bigjoin \{ \absmem(\absloc) \mid \absloc \in
      \referabssem{\refer}(\abselem) \}
    }
    {
      \exprabssem{\refer}(\abselem) = \absval
    }

    \inferrule*[width=0.48\textwidth]
    {
      \exprabssem{\expr_1}(\abselem) = \absval_1\\
      \cdots\\
      \exprabssem{\expr_n}(\abselem) = \absval_n
    }
    {
      \exprabssem{\op(\expr_1, \cdots, \expr_n)}(\abselem) = \ops(\absval_1,
      \cdots, \absval_n)
    }
  \end{mathpar}
  \vspace*{-1.5em}
  \caption{The semantics of view transition for the core language of JavaScript}
  \vspace*{-.5em}
  \label{fig:core-view-trans}
\end{figure*}

\[
  \begin{array}{ll@{~}c@{~}l}
    \text{Programs} & \prog &::=& (\lab: \inst)^*\\

    \text{Labels} & \lab &\in& \labset\\

    \text{Instructions} & \inst &::=&
    \refer = \expr \mid
    \refer = \kwobj \mid
    \refer = \expr ( \expr ) \mid
    \kwret \; \expr \mid
    \kwif \; \expr \; \lab\\

    \text{References} & \refer &::=&
    x \mid
    \expr [ \expr ]\\

    \text{Expressions} & \expr &::=&
    \pval \mid
    \lambda x. \; \lab \mid
    \refer \mid
    \op(\expr^*)\\
  \end{array}
\]

A program $\prog$ is a sequence of labeled instructions. An instruction $\inst$
is an expression assignment, an object creation, a function call, a return
instruction, or a branch.  A reference $\refer$ is a variable or a property
access of an object.  An expression $\expr$ is a primitive, a lambda function, a
reference, or an operation between other expressions.
\[
  \begin{array}{lr@{~}c@{~}l@{~}c@{~}l}
    \text{States} & \st &\in& \stset &=& \labset \times \memset \times
    \ctxtset \times \eaddrset\\
    \text{Memories} & \mem &\in& \memset &=& \locset \finmap \valset\\
    \text{Contexts} & \ctxt &\in& \ctxtset &=& \eaddrset \finmap (\eaddrset
    \times \labset \times \locset)\\
    \text{Locations} & \loc &\in& \locset &=& (\eaddrset \times \varset) \uplus
    (\oaddrset \times \strset)\\
    \text{Values} & \val &\in& \valset &=& \pvalset \uplus \oaddrset \uplus
    \fvalset\\
    \text{Primitives} & \pval &\in& \pvalset &=& \strset \uplus \cdots\\
    \text{Addresses} & \addr &\in& \addrset &=& \eaddrset \uplus \oaddrset\\
    \text{Functions} & \fval{x}{\lab} &\in& \fvalset &=& \varset \times
    \labset\\
  \end{array}
\]

States $\stset$ consist of labels $\labset$, memories $\memset$, contexts
$\ctxtset$, and environment addresses $\eaddrset$.  A memory $\mem \in \memset$
is a finite mapping from locations to values.  A context $\ctxt \in \ctxtset$ is
a finite mapping from environment addresses to tuple of environment addresses,
return labels, and left-hand side locations.  A location $\loc \in \locset$ is a
variable or an object property; a variable location consists of an environment
address and its name, and an object property location consists of an object
address and a string value.  A value $\val \in \valset$ is a primitive, an
address, or a function value.  An address $\addr \in \addrset$ is an environment
address or an object address.  A function value $\fval{x}{\lab} \in \fvalset$
consists of a parameter name and a body label.  In the core language, the closed
scoping is used for functions for brevity, thus only parameters and local
variables are accessible in a function body.

We formulate the concrete semantics of the core language as described in
Figure~\ref{fig:core-trans-rel}.  The transition
relation between concrete states is defined with the semantics of references and
expressions using two different forms \fbox{$\referrule{\st}{\refer}{\loc}$} and
\fbox{$\exprrule{\st}{\expr}{\val}$}, respectively.  The initial states are
$\istset = \{ (\ilab, \varnothing, \epsilon, \tladdr) \}$ where $\ilab$ denotes
the initial label, $\epsilon$ empty map, and $\tladdr$ the top-level environment
address.  The function $\labnext$ returns the next label of a given label in the
current program $\prog$.

\subsection{Abstract Semantics}
In the abstract semantics of the core language, we use the flow sensitivity with a
flow sensitive view abstraction $\fsviewmap: \labset \rightarrow \dom$ that
discriminates states using their labels: $\forall \lab \in \labset. \;
\fsviewmap(\lab) = \{ \st \in \stset \mid \st = (\lab, \_, \_, \_) \}$. Thus, the
sensitive abstract domain is defined as $\sabsdom = \labset \rightarrow
\absdom$.  We define an abstract state $\abselem \in \absdom$ as a tuple of an
abstract memory, an abstract context, an abstract address, and an
abstract counter as follows:
\[
  \begin{array}{l@{~}r@{~}c@{~}l@{~}c@{~}l}
\text{Abstract states} & \abselem &\in& \absdom &=& \absmemset \times \absctxtset
\times \absaddrset \times \abscountset\\
\text{Abstract memories} & \absmem &\in& \absmemset &=& \abslocset \finmap
\absvalset\\
\text{Abstract locations} & \absloc &\in& \abslocset &=& (\absaddrset \times
\varset) \uplus (\absaddrset \times \strset)\\
\text{Abstract addresses} & \absaddr &\in& \absaddrset &=& \labset\\
\text{Abstract contexts} & \absctxt &\in& \absctxtset &=& \absaddrset \finmap
\powerset{\absaddrset \times \viewset \times \powerset{\abslocset}}\\
\text{Abstract counters} & \abscount &\in& \abscountset &=& \absaddrset
\rightarrow \{ \abszero, \absone, \absmany \}\\
\text{Abstract values} & \absval &\in& \absvalset &=& \powerset{\pvalset
\uplus \absaddrset \uplus \fvalset}\\
  \end{array}
\]

An abstract memory $\absmem \in \absmemset$ is a finite mapping from abstract
locations $\abslocset$ to abstract values $\absvalset$.  Abstract locations
$\abslocset$ are pairs of abstract addresses with variable names or string
values. Abstract addresses $\absaddrset$ are defined with the
\textit{allocation-site abstraction} that partitions concrete addresses
$\addrset$ based on their allocation sites $\labset$.  Abstract contexts
$\absctxtset$ are finite maps from abstract addresses to powersets of triples of
abstract addresses, views, and powerset of abstract locations.  For abstract
counting~\cite{abstract-gc-counting, revisit-recency} in static analysis, we
define abstract counters $\abscountset$ that are mappings from abstract addresses to
their abstract counts representing how many times each abstract address has been
allocated; $\abszero$ denotes that it has never been allocated, $\absone$ once,
and $\absmany$ more than or equal to twice.

We define the semantics of the view transition for the core language.  For abstract
memories, we use the notation $\absmem[L \mapstos \absval]$ to represent the
update of multiple abstract locations in $L$ with the abstract value $\absval$.
It performs the strong update if the abstract address for an abstract location
$(\absaddr, \_) \in L$ is singleton: $\abscount(\absaddr) = \absone$.
Otherwise, it performs the weak update for the analysis soundness.  We use
the increment function $\inc: \abscountset \times \absaddrset \rightarrow
\abscountset$ of the abstract counter defined as follows:
\[
  \inc(\abscount)(\absaddr_0) = \lambda \absaddr \in \absaddrset. \; \left\{
    \begin{array}{ll}
      \absone & \text{if} \; \absaddr = \absaddr_0 \wedge
      \abscount(\absaddr_0) = \abszero\\
      \absmany & \text{if} \; \absaddr = \absaddr_0 \wedge
      \abscount(\absaddr_0) = \absone\\
      \abscount(\absaddr) & \text{otherwise}
    \end{array}
  \right.
\]

\subsection{{\SealeD} Execution}

We define {\sealed} states by not only extending the concrete values
$\valset$ with {\sealed} values $\symbset$ but also adding the abstract counters~
$\abscountset$:
\[
  \begin{array}{r@{~}c@{~}l}
    \symbstset &=& \labset \times \memset \times \ctxtset \times \eaddrset
    \times \abscountset\\
    \ctxtset &=& \eaddrset \finmap ((\eaddrset \times \labset \times \locset)
    \uplus \symbset)\\
    \valset &=& \pvalset \uplus \oaddrset \uplus \fvalset \uplus \symbset\\
    \abscountset &=& \oaddrset \rightarrow \{ \abszero, \absone, \absmany \}\\
  \end{array}
\]
Because JavaScript provides open objects, the properties of objects can be dynamically added or deleted.
Moreover, since object properties are string values that can be constructed at run time,
it is difficult to perform sound strong updates in static analysis.
To check the possibility of strong updates during {\sealed} execution,
we augment its states with the abstract counters $\abscountset$.

For each abstract value in a given abstract state,
if the abstract value denotes a single concrete value,
the converter $\asconverter: (\viewset \times
\absdom) \rightarrow (\absimapset \times \symbstset)$
keeps it; otherwise, $\asconverter$ replaces the abstract
value with its unique identifier and maintains the mapping from the
unique identifier to the abstract value to construct an abstract instantiation map.
The opposite converter $\saconverter: (\absimapset \times
  \symbstset)  \rightarrow (\viewset \times \absdom)$
recovers abstract values from their unique identifiers using the abstract instantiation map.
We define the {\sealed} transition relation $\symbtrans$
only if the next step does not require actual values of any {\sealed} values.
Otherwise, a given {\sealed} state does not have any {\sealed}
transitions to apply.  For example, we add the following rule:
\begin{mathpar}
  \inferrule
  {
    \prog(\lab) = \kwret \; \expr\\
    \exprrule{\symbst}{\expr}{\val}\\
    \ctxt(\addr) \in \symbset
  }
  {
    \symbst = (\lab, \mem, \ctxt, \addr, \abscount)\
    \symbtrans \excst
  }
\end{mathpar}
for the $\kwret$ statement. We extend each rule of the concrete semantics
to support such behaviors of {\sealed} values.

%% file: implementation.tex
\section{Implementation}\label{sec:implementation}
We implemented JavaScript static analysis using dynamic shortcuts
presented in Section~\ref{sec:javascript} in a prototype implementation dubbed
$\tool$.  The tool is an extension of an existing state-of-the-art JavaScript
static analyzer SAFE~\cite{safe, safe2} with a dynamic analyzer
Jalangi~\cite{jalangi}, and it is an open-source project and available
online~\footnote{https://github.com/kaist-plrg/safe-ds}.  In this
section, we introduce challenges and solutions in implementing dynamic
shortcuts on existing JavaScript analyzers.

\paragraph{{\SealeD} Values.}
The main challenge of implementing dynamic shortcuts is to support {\sealed} execution on an existing JavaScript engine.  To represent an abstract
value, we use the \jscode{Proxy} object introduced in ECMAScript 6
(2015, ES6)~\cite{es6}, which allows developers to handle internal behaviors
of specific objects such as property reads and writes and implicit conversions.
We are inspired by \textsc{Mimic}~\cite{mimic}, which used \jscode{Proxy} to
capture accesses from internals of opaque functions.  When the dynamic analyzer
constructs an execution environment at the start of a dynamic shortcut, it
creates \jscode{Proxy} objects to represent abstract values via the
following \jscode{getSealedValue} function:
\begin{lstlisting}[style=myJSstyle]
function getSealedValue() {
  function detect() { /* access detection */ }
  return new Proxy(function() {}, {
    getPrototypeOf: detect,  ...
    construct     : detect
  }); }
var x = getSealedValue();
var y = x;
var z = x + 1;
\end{lstlisting}
The function creates a sealed value as a proxy object with a dummy
function object and a handler for all 13 traps using an access detection
function \jscode{detect}.  A sealed value invokes the function \jscode{detect}
when any of 13 pre-defined traps are operated on the object, which enables us to
determine whether an object is sealed or not.  For example, the variable
\jscode{y} successfully points to the same sealed value stored in \jscode{x}, but the
program invokes the function \jscode{detect} on line 9 because \jscode{x + 1} requires
the actual value of the sealed value.  In addition, we instrument unary and binary
operations in Jalangi so that we can detect all the accesses on the
sealed value beyond the 13 traps provided by \jscode{Proxy}.
Using this idea, we successfully extended the
JavaScript engine to support {\sealed} execution.

\paragraph{Synchronization of Control Points.}
For seamless interaction between static analysis and {\sealed} execution,
synchronization of control points in both sides is necessary.
The SAFE static analyzer and the Jalangi dynamic analyzer
have their own notations for control points that are not directly
compatible.
We use the source-code location of a target program as a key to synchronize.
Even though they use different parsers and we faced numerous location mismatches for corner cases,
we could synchronize control points of two analyzers by using the closest match
of their source-code locations rather than using their exact match.

\paragraph{Function-Level Dynamic Shortcut.}
A dynamic shortcut is activated when the current abstract state passes the
filter $\checker$.  Because SAFE and Jalangi are implemented in different languages, Scala and JavaScript,
respectively, we represent abstract states as JSON
objects and communicate between analyzers by passing JSON objects through a localhost
server.  If the filter admits dynamic shortcuts generously, the
analysis may suffer from frequent communications between static and dynamic
analyzers.  To adjust such a burden, $\tool$ supports only \textit{function-level} dynamic
shortcuts by activating dynamic shortcuts in function entries and deactivating them
in their corresponding function exits.

\paragraph{Termination.}
To guarantee the termination of static analysis using dynamic shortcuts, the
converter $\asconverter$ should pass an analysis element $(\view, \abselem)$
only when it terminates in a time bound $N$.  Since statically checking the
termination property is difficult, we simply perform {\sealed}
execution with a pre-determined time limit of 5 seconds.
When it times out, we treat it as a failure in conversion;
otherwise, we use the result of {\sealed} execution.

%% file: eval.tex
\section{Evaluation}\label{sec:eval}

We evaluate $\tool$ using the following research questions:
\begin{itemize}
\item \textbf{RQ1) Analysis Speed-up:} How much analysis time is reduced by
using dynamic shortcuts?
\item \textbf{RQ2) Precision Improvement:} How much analysis precision is
improved by using dynamic shortcuts?
\item \textbf{RQ3) Opaque Function Coverage:} How many opaque functions are
covered only by dynamic shortcuts?
\end{itemize}
We selected the official 306 tests of Lodash 4
(v.4.17.20)\footnote{https://github.com/lodash/lodash/blob/4.17.20/test/test.js}
used in the examples in Section~\ref{sec:motivation} as our evaluation target.
Recent work~\cite{value-refinement,
value-partitioning} also used the tests to evaluate their techniques.
Among them, we filtered out 37 tests that use JavaScript language
features SAFE does not support such as dynamic code generation using
\njscode{Function}, getters and setters, and browser-specific features like $\jscode{__proto__}$.
Thus, we used 269 out of 306 tests for the evaluation of $\tool$ and compared
its evaluation results with those of the baseline analyzer, SAFE.
For both SAFE and $\tool$, we used 400-depth, 10-length loop strings and
30-length call strings for precise analysis, and added some incomplete
models for opaque functions to soundly analyze Lodash tests.
We performed our experiments on a Ubuntu machine
equipped with 4.2GHz Quad-Core Intel Core i7 and 32GB of RAM.

\subsection{Analysis Speed-up}

\begin{figure}[t]
  \centering
  \vspace{2mm}
  \includegraphics[width=\linewidth]{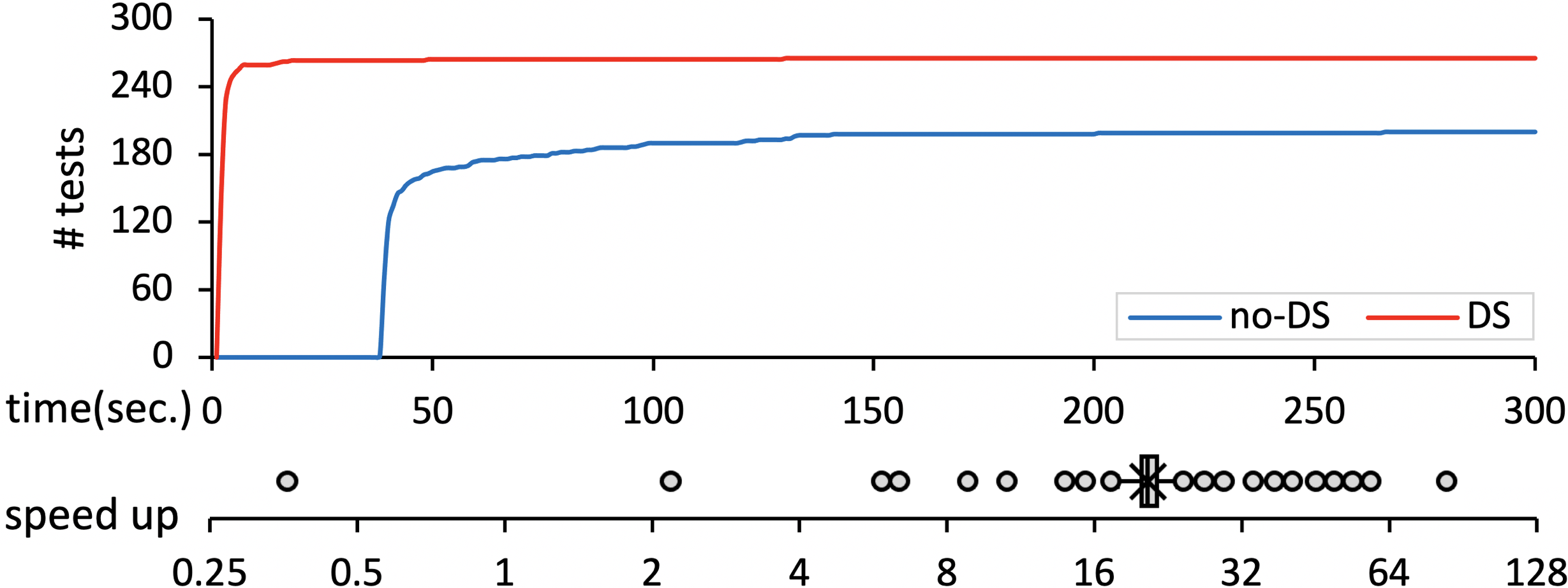}
  \vspace*{-1.5em}
  \caption{Analysis time for Lodash 4 \textit{original} tests without (no-DS)
  and with (DS) dynamic shortcuts within 5 minutes}
  \label{fig:conc-analysis-time}
\vspace*{-.5em}
\end{figure}

\begin{figure}[t]
  \centering
  \includegraphics[width=\linewidth]{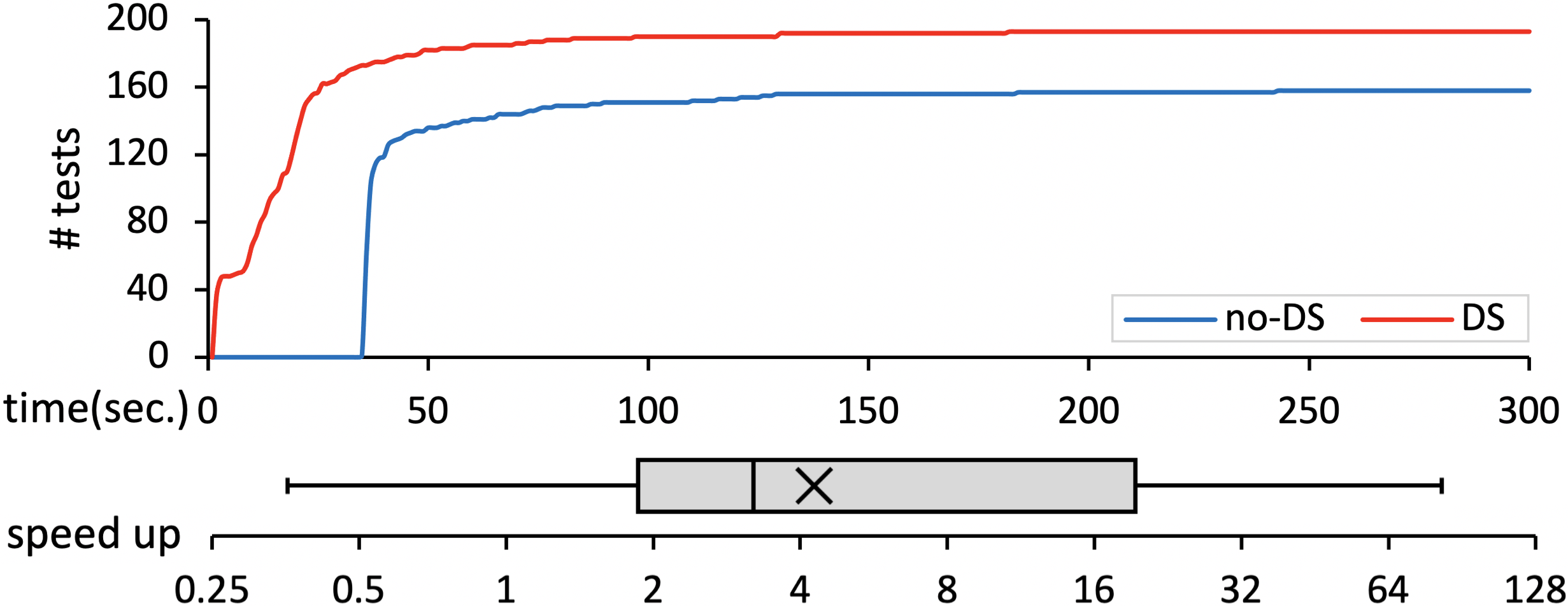}
  \vspace*{-1.5em}
  \caption{Analysis time for Lodash 4 \textit{abstracted} tests without (no-DS)
  and with (DS) dynamic shortcuts within 5 minutes}
  \label{fig:abs-analysis-time}
  \vspace*{-1em}
\end{figure}

We evaluated the effectiveness of dynamic shortcuts by static
analysis of 269 Lodash 4 tests with and without dynamic shortcuts.
Figure~\ref{fig:conc-analysis-time} depicts cumulative distribution charts for
their analysis time and a box plot in a logarithmic scale for speed up after
applying dynamic shortcuts.  In the upper chart, the $x$-axis is time and the
$y$-axis shows the number of tests within the time.  While the baseline analysis
(no-DS) finished analysis of 200 out of 269 tests within 5 minutes, our tool
(DS) finished analysis of 265 tests using dynamic shortcuts.  For finished
tests, the average analysis time is 49.46 seconds for no-DS and 3.21 seconds for
DS.  Among 200 tests analyzed by no-DS, one test is timeout in DS, thus
199 tests are analyzable by both analyzers. For them, we depict the box plot for
analysis speed up by dynamic shortcuts.  It shows that DS
outperforms no-DS up to 83.71$\x$ and 22.30$\x$ on
average.  Only for one test using $\jscode{_.sample}$, which
randomly samples a value from a given array, DS showed
0.36$\x$ speed of no-DS due to 24 times uses of dynamic shortcuts.

Note that since most tests use concrete values instead of
non-deterministic inputs, they can be analyzed by a few number of dynamic shortcuts.
In fact, among 269 tests, 259 tests are analyzed
by a single dynamic shortcut without using abstract semantics.
However, in real-world JavaScript programs, arguments of library
functions may include non-deterministic inputs.
To evaluate $\tool$ in a real-world setting,
we modified the tests to use abstract values.
We made abstract values by randomly selecting literals and replacing
one of them with its corresponding abstract value.
For example, if we select a numeric literal \jscode{42}, we modified it to the abstract numeric value
$\top_{\code{num}}$, which represents all the numeric values.
In the remaining section, we evaluated $\tool$ using the \textit{original} tests
and the \textit{abstracted} tests.

For abstracted tests as well, DS outperformed no-DS.
Figure~\ref{fig:abs-analysis-time} shows the analysis time of the abstracted tests.
Among 269 abstracted tests, no-DS finished analysis of 158 tests within 5 minutes,
but DS finished analysis of 193 tests.  For finished tests, the average analysis
time is 44.88 seconds for no-DS and 19.05 seconds for DS. Among 158 tests analyzed by no-DS, DS
timed-out for 2 tests.  For 156 tests analyzable by both analyzers,
DS outperformed no-DS up to 78.07$\x$ and 7.81$\x$ on average.
Except for 9 test cases, using dynamic shortcuts did show speed-ups.

\begin{figure}[t]
  \centering
  \vspace{2mm}
  \includegraphics[width=\linewidth]{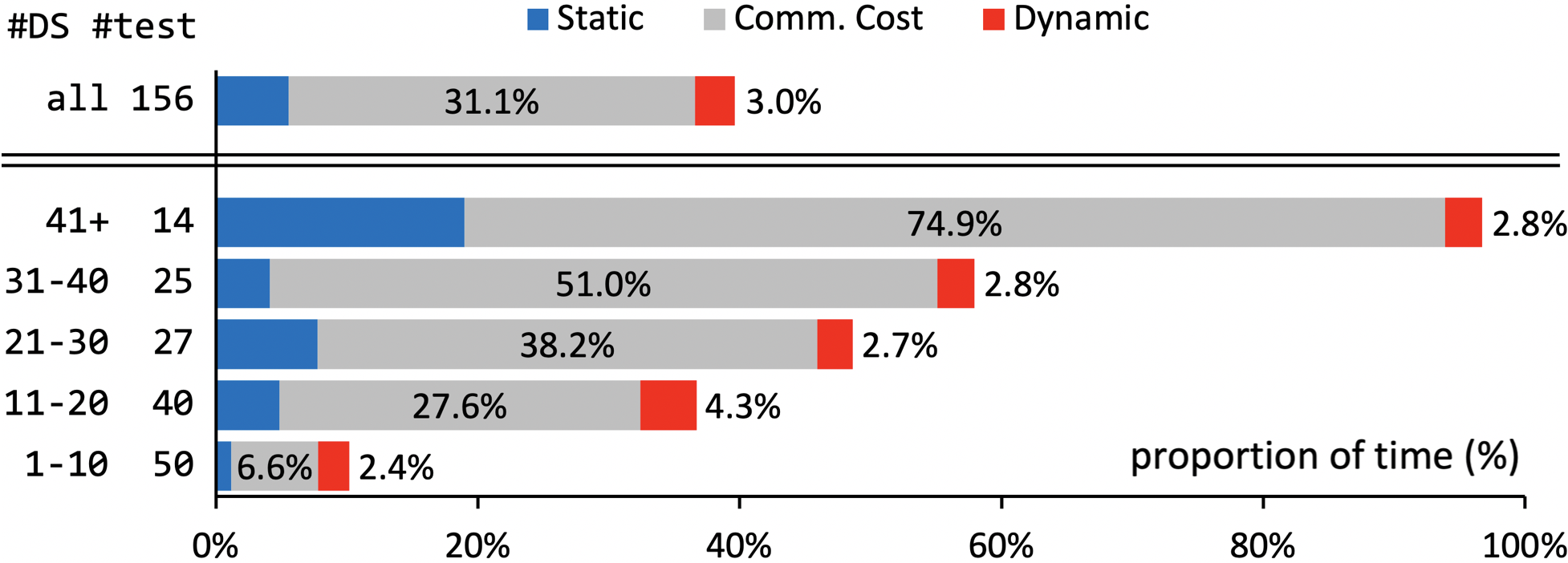}
  \vspace*{-1.5em}
  \caption{Analysis time ratio for 156 \textit{abstracted} tests}
  \label{fig:abs-analysis-ratio}
  \vspace*{-.5em}
\end{figure}

Unlike for the original tests, analysis of 156 abstracted tests invoked
20.35 dynamic shortcuts.  Because taking a dynamic shortcut
requires conversion between abstract states and {\sealed} values
and their exchanges between the static analyzer and the dynamic analyzer,
using dynamic shortcuts multiple times may incur more performance
overhead than performance benefits by using {\sealed} execution.
One conjecture is that the communication cost between the static
analyzer and the dynamic analyzer may be proportional to the number of
dynamic shortcuts.

To experimentally evaluate the conjecture, we investigated the relationship between
the communication cost (Comm. Cost) between analyzers and the number of dynamic shortcuts.
For 199 original tests, Comm. Cost was only
1.58\% compared to the analysis time of no-DS.  However, for 156
abstracted tests, Comm. Cost was 31.06\% compared to the analysis
time of no-DS.  Figure~\ref{fig:abs-analysis-ratio} presents the
analysis time ratio for 156 abstracted tests.
The $x$-axis represents the time ratio normalized by the total analysis time of
no-DS and the $y$-axis denotes the number of dynamic
shortcuts and the number of corresponding tests.
For all 156 tests, Comm. Cost is larger than
both the static analysis time (Static) and the dynamic analysis
time (Dynamic).  When dynamic shortcuts are performed less than 10 times,
Comm. Cost is modest compared to the baseline static
analysis time.  However, the more dynamic shortcuts are performed,
the less the performance benefits by using dynamic shortcuts.
Specifically, when dynamic shortcuts are performed more than 30 times,
Comm. Cost is even larger than half of cost of no-DS.
Based on this evaluation result, we believe that we can leverage
dynamic shortcuts by optimizing Comm. Cost between
the static analyzer and the dynamic analyzer.  One possible approach is to
reduce the sizes of JSON objects that represent abstract and sealed states by
representing only their updated parts.  Another approach could be to use a
communication system faster than a localhost server for passing JSON objects.

\begin{figure}[t]
  \centering
  \begin{subfigure}[t]{0.43\columnwidth}
    \includegraphics[width=\textwidth]{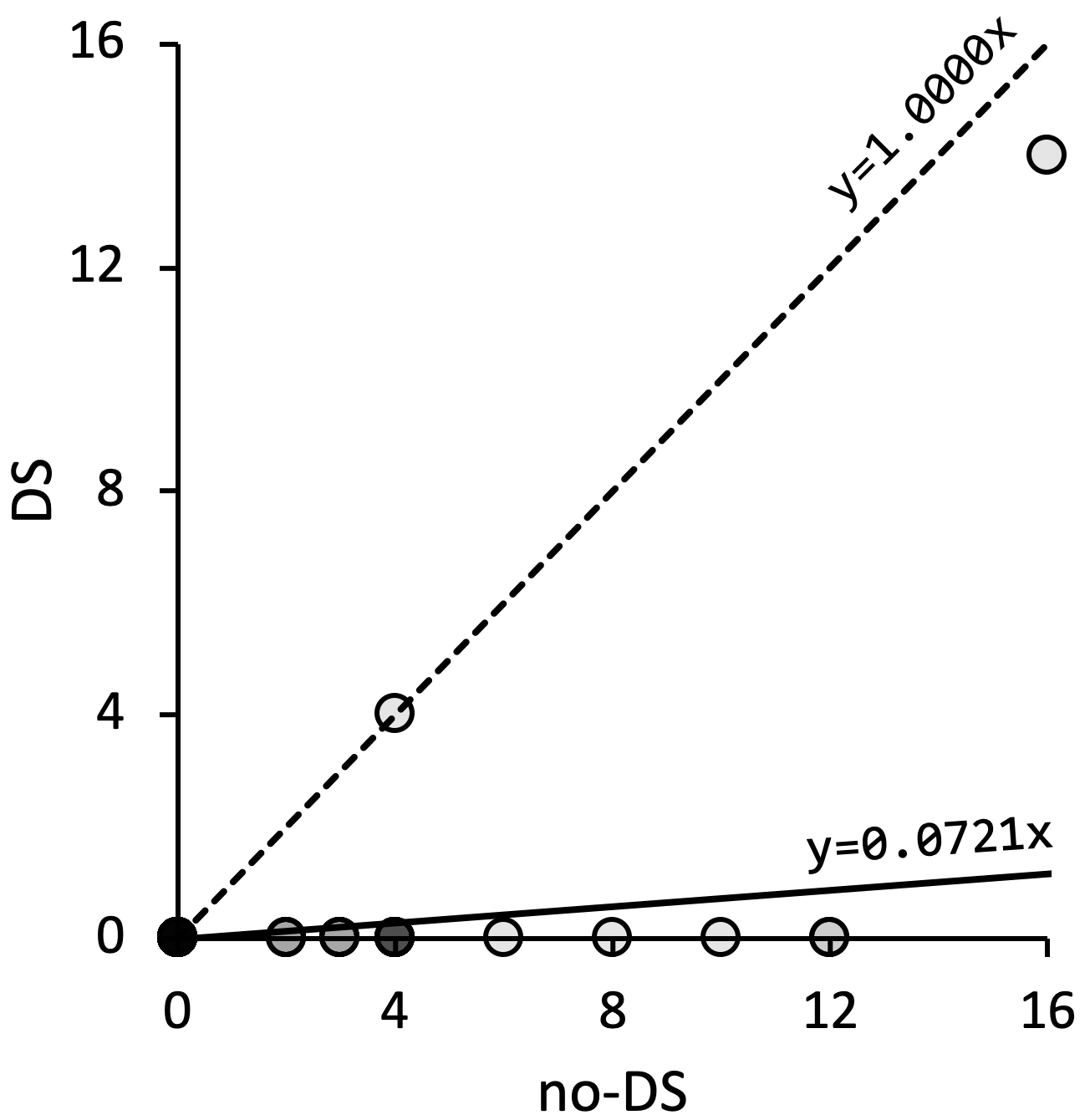}
    \caption{199 \textit{original} tests}
    \label{fig:conc-precision}
  \end{subfigure}
  \begin{subfigure}[t]{0.43\columnwidth}
    \includegraphics[width=\textwidth]{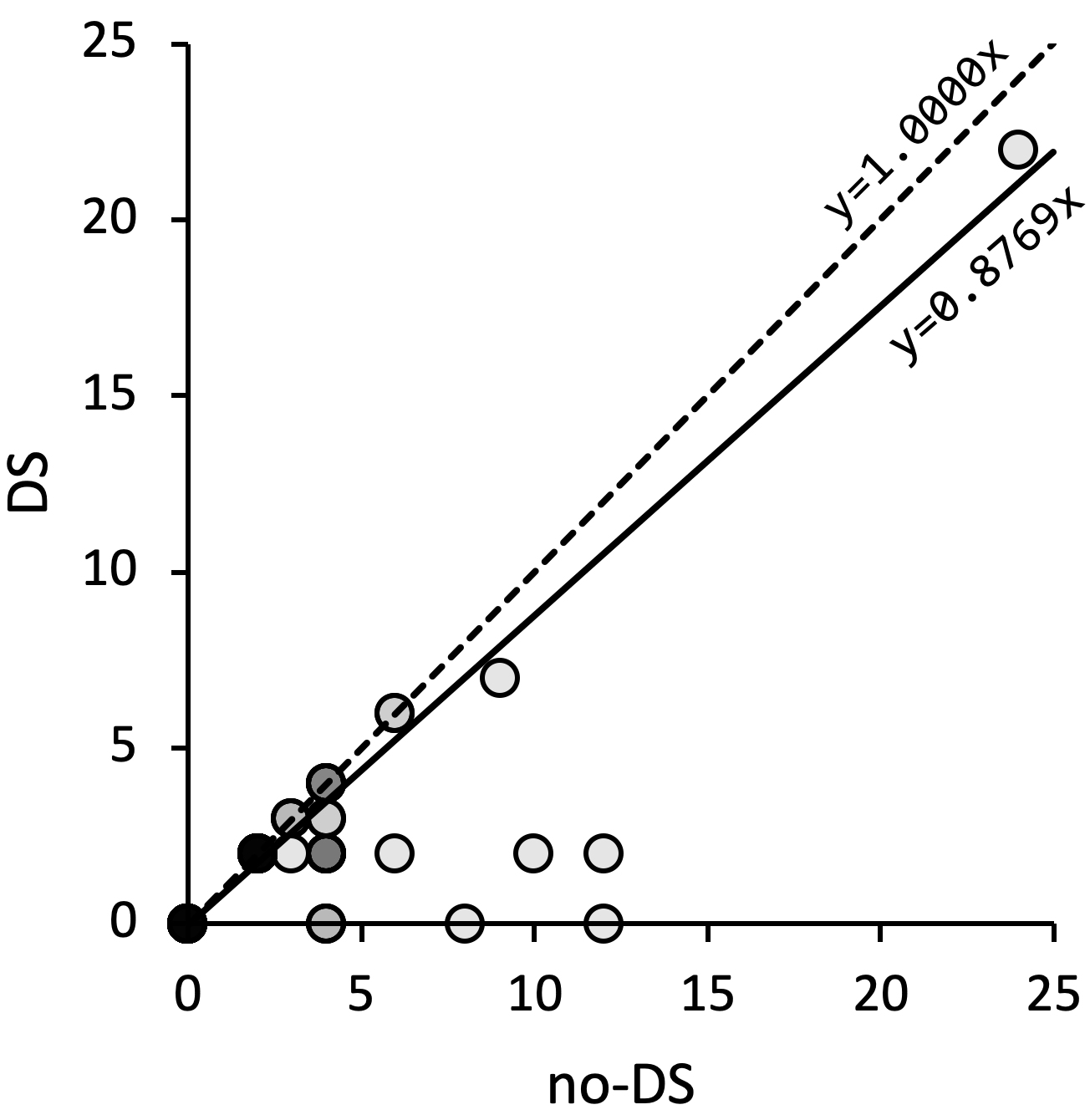}
    \caption{156 \textit{abstracted} tests}
    \label{fig:abs-precision}
  \end{subfigure}
  \vspace*{-1em}
  \caption{Failed assertions of analysis without (no-DS) and with (DS) dynamic shortcuts}
  \vspace*{-1.5em}
  \label{fig:precision}
\end{figure}

\begin{table*}[t]
  \vspace{2mm}
  \caption{Number of original (orig.) and abstracted (abs.) tests using dynamic shortcuts
only for each JavaScript built-in library}
  \label{table:func-replace}
  \vspace*{-1.5em}
  \centering
  \scriptsize
  \[
    \qquad
    \qquad
    \quad
    \begin{array}{c|l|c|c?c|l|c|c?c|l|c|c}

      \myhead{Object}       {Function}        {\# Replaced}

      \mysuff{1}{}          {Array    }  {204 / 205}{119 / 141} 	& \mysutf{1}{}             {String     }  { 20 /  20}{ 13 /  14} 	& \mysutf{1}{}       {Object        }  {265 / 265}{181 / 193} \mylinefff
      \mysuff{1}{}          {new Array}  {  0 /   0}{  0 /   7} 	& \mysuff{1}{}             {toString   }  {  0 /   0}{  0 /  14} 	& \mysutf{1}{}       {getPrototypeOf}  { 56 /  56}{ 34 /  35} \mylinefff
      \mysuff{1}{}          {isArray  }  {264 / 265}{181 / 193} 	& \mysuff{1}{}             {valueOf    }  {  0 /   0}{  0 /  20} 	& \mysbtt{1}{}       {create        }  {265 / 265}{193 / 193} \mylinefff
      \mysutf{1}{}          {concat   }  {265 / 265}{189 / 193} 	& \mysutt{1}{}             {charAt     }  {  8 /   8}{  6 /   6} 	& \mysutf{2}{Object} {defineProperty}  {265 / 265}{190 / 193} \mylinefff
      \mysutt{1}{}          {join     }  {265 / 265}{193 / 193} 	& \mysutt{1}{}             {charCodeAt }  { 15 /  15}{  8 /   8} 	& \mysbtt{1}{}       {freeze        }  {  1 /   1}{  1 /   1} \mylinefff
      \mysutt{1}{}          {pop      }  { 25 /  25}{ 14 /  14} 	& \mysutt{1}{}             {indexOf    }  {  2 /   2}{  1 /   1} 	& \mysutf{1}{}       {keys          }  {265 / 265}{191 / 193} \mylinefff
      \mysutf{2}{Array}     {push     }  {265 / 265}{186 / 193} 	& \mysutf{2}{String}       {match      }  { 26 /  26}{ 16 /  18} 	& \mysuff{1}{}       {toString      }  {264 / 265}{138 / 193} \mylinefff
      \mysutt{1}{}          {reverse  }  { 10 /  10}{  6 /   6} 	& \mysutf{1}{}             {replace    }  { 56 /  56}{ 31 /  37} 	& \mysutf{1}{}       {hasOwnProperty}  {265 / 265}{190 / 193} \mylinefft
      \mysutt{1}{}          {shift    }  {  3 /   3}{  2 /   2} 	& \mysutf{1}{}             {slice      }  {265 / 265}{191 / 193} 	& \mysbtt{1}{JSON}   {stringify}       {  1 /   1}{  1 /   1} \mylinefft
      \mysutt{1}{}          {slice    }  {265 / 265}{193 / 193} 	& \mysutt{1}{}             {split      }  {  5 /   5}{  2 /   2} 	& \mysutf{1}{}       {parseInt}        {  2 /   2}{  1 /   2} \mylinefff
      \mysutf{1}{}          {sort     }  { 69 /  69}{ 38 /  39} 	& \mysutf{1}{}             {substring  }  {214 / 214}{136 / 145} 	& \mysutf{1}{Global} {isNaN   }        { 15 /  15}{ 11 /  40} \mylinefff
      \mysutf{1}{}          {splice   }  { 25 /  25}{  9 /  12} 	& \mysutf{1}{}             {toLowerCase}  {215 / 215}{135 / 146} 	& \mysbtt{1}{}       {isFinite}        {  3 /   3}{  1 /   1} \mylinefft
      \mysutt{1}{}          {unshift  }  {  2 /   2}{  2 /   2} 	& \mysutf{1}{}             {toUpperCase}  { 11 /  11}{  6 /   7} 	& \mysbtt{1}{}       {RegExp    }      {265 / 265}{193 / 193} \mylinefff
      \mysutf{1}{}          {indexOf  }  { 94 /  94}{ 61 /  66} 	& \mysutt{1}{}             {fromCharCode} {  1 /   1}{  1 /   1}  &	\mysuff{2}{RegExp} {new RegExp}      {  0 /   0}{  0 /   1} \mylineftf
      \mysutf{1}{}          {every    }  { 92 /  92}{ 43 /  47} 	& \mysuff{1}{Date}         {new Date}     {  0 /   1}{  0 /   1} 	& \mysbtt{1}{}       {exec      }      {265 / 265}{193 / 193} \mylinettf
      \mysuff{1}{}          {ceil }      { 37 /  38}{ 20 /  21} 	& \mysutt{1}{}             {Number}       {  2 /   2}{  2 /   2} 	& \mysuff{1}{}       {test      }      {264 / 265}{185 / 193} \mylinefft
      \mysuff{1}{}          {floor}      { 16 /  18}{  8 /  10} 	& \mysutf{1}{Number}       {toFixed}      {  1 /   1}{  0 /   0} 	& \mysutf{1}{}       {Error     }      {  1 /   1}{  0 /   1} \mylinefff
      \mysuff{2}{Math}      {max  }      {264 / 265}{179 / 193} 	& \mysuff{1}{}             {valueOf}      {  0 /   0}{  0 /  28} 	& \mysuff{1}{Error}  {new RangeError}  {  0 /   0}{  0 /   2} \mylineftf
      \mysutf{1}{}          {min  }      { 64 /  64}{ 31 /  44} 	& \mysutt{1}{}             {toString}     {265 / 265}{193 / 193} 	& \mysuff{1}{}       {new TypeError }  {  0 /   0}{  0 /   7} \mylinefft
      \mysutt{1}{}          {pow  }      { 11 /  11}{  6 /   6} 	& \mysuff{1}{Function}     {apply   }     {263 / 265}{133 / 193} 	& \mysbtt{2}{Boolean}{Boolean}         {  3 /   3}{  2 /   2} \mylinefff
      \mysutt{1}{}          {round}      {  2 /   2}{  1 /   1} 	& \mysuff{1}{}             {call    }     {259 / 265}{ 50 / 193} 	& \mysuff{1}{}       {valueof}         {  0 /   0}{  0 /   7}
    \end{array}
  \]
\end{table*}

\subsection{Precision Improvement}

To evaluate the analysis precision improvement of dynamic shortcuts, we measured
the number of failed assertions produced by no-DS and DS.  Because both no-DS
and DS are sound, high (low) number of failed assertions denotes low (high)
analysis precision.

Figure~\ref{fig:precision} depicts the comparison of the analysis precision
between no-DS and DS.  The $x$-axis and the $y$-axis denote the number of failed
assertions produced by no-DS and DS, respectively.  For example, if both DS and
no-DS failed 4 assertions in an original test, the figure shows a circle at the
point (4, 4).  Since multiple circles can be at the same point if both
DS and no-DS failed the same number of assertions, we use darker
gray to denote a larger number of tests in a heat-map form. The darker the
circle is, the more tests it indicates.  The dotted line denotes the $y=x$ line
and all the circles are below or on the line, which means DS produces less or equal numbers
of assertions compared to no-DS for all tests.  On the other hand, the solid
line denotes the average improvement, which is the ratio of the total number of
failed assertions produced by no-DS to that produced by DS.  For 199 original
tests that are analyzable by both analyzers, Figure~\ref{fig:precision}(a) shows
that dynamic shortcuts reduced the number of failed assertions by 92.79\% on
average.  For 156 abstracted tests that are analyzable by both analyzers,
Figure~\ref{fig:precision}(b) shows that dynamic shortcuts successfully cut down
the number of failed assertions by 12.31\% on average.  Thus, on average,
dynamic shortcuts removed analysis of 92.79\% and 12.31\% failed assertions for
original and abstracted tests, respectively.

\vspace*{-1em}

%% file: eval-opaque.tex
\subsection{Opaque Function Coverage}
To evaluate how much manual modeling efforts of opaque functions
are reduced by dynamic shortcuts, we measured the number of tests for which
opaque functions are analyzed only by dynamic analysis not by
static analysis.  Table~\ref{table:func-replace} summarizes the result.
For 265 original tests and 193 abstracted tests that DS finished analysis, we measured the
number of tests that use only dynamic shortcuts instead of manual modeling
for each JavaScript built-in library function.  For each row,
\textbf{Object} column denotes a built-in object, \textbf{Function} a function
name, and \textbf{\#~Replaced} the number of tests successfully replacing manual
modeling via dynamic shortcuts over the total number of tests using the target function.
For example, the first row in the leftmost side describes that \jscode{Array} is used in
205 original tests and 141 abstracted tests.  Among them, 204 original
tests and 119 abstracted tests are successfully analyzed by using dynamic shortcuts
instead of manual modeling of \jscode{Array}.  Each filled cell describes
a fully replaceable case.  Therefore, dynamic
shortcuts effectively lessen the burden of manual modeling for JavaScript
built-in functions. For the original tests, 45 out of 63 built-in functions are replaceable
for them.  For the abstracted tests, 22 built-in functions are analyzed by only dynamic shortcuts.

%% file: related.tex
\section{Related Work}\label{sec:related}

\paragraph{Combined Analysis}
The most related previous work is combined analysis that utilizes dynamic
analysis during Java static analysis introduced by \citet{concerto}.
They proved that their combined analysis is sound and showed that it could
significantly improve the precision and performance of Java static analysis by
evaluating their tool, \concerto.  However, their approach has several
limitations compared with dynamic shortcuts.  First, it syntactically
divides a given program to \textit{applications} parts for static
analysis and \textit{frameworks} parts for dynamic analysis.  Thus, it cannot
freely switch between static analysis and dynamic analysis.  It is
even impossible to perform both static and dynamic analysis of 
the same program part in different contexts.  In addition, while they
introduced \textit{mostly-concrete interpretation} similar to our
{\sealed} execution, it supports only a special \textit{unknown}
value that represents any possible value.  Thus, it cannot preserve
the precision of complex abstract domains~\cite{revisit-recency,
regex, weaklyAPLAS, weaklySPE} frequently used in JavaScript static analysis.
On the contrary, {\sealed} execution automatically detects when
to switch to static analysis to use abstract semantics for abstract values.
Finally, \concerto\ preserves the soundness when a program satisfies
the \textit{state separation hypothesis}.  It assumes that the states
of application parts and framework parts are not interrogated
or manipulated by each other.  While the assumption may be reasonable
for static analysis of Java applications using external libraries, it
is not satisfied for JavaScript programs in general.  Unlike their
approach, our approach does not have any assumptions between static and
dynamic analysis parts.

\paragraph{Concolic Execution}
Concolic execution~\cite{dart} is closely related to dynamic shortcuts because
it also leverages concrete execution for symbolic execution.
Symbolic execution~\cite{symbolic} is an execution of a program with symbolic values, and
it can be treated as an abstract interpretation with symbolic expressions and path constraints.
To resolve path constraints with symbolic expressions,
symbolic execution engines such as KLEE~\cite{klee} and SAGE~\cite{sage} utilize
Satisfiability Modulo Theory (SMT) solvers as back-end modules.
On the contrary, we formalized dynamic shortcuts as a technique to combine concrete execution with
a general abstract interpretation, not only with symbolic execution.
Thus, dynamic shortcuts are theoretically applicable to any kind of abstract interpretation,
including symbolic execution, and it is a more general definition of concolic execution.

\paragraph{Automatic Modeling}
For static analysis of JavaScript programs, modeling behaviors of built-in
libraries or host-dependent functions is necessary because they are opaque code.
Since manual modeling is error-prone and labor-intensive,
researchers~\cite{safewapi, safets} have utilized type information to
automatically model their behaviors.  However, type is not enough
to reflect complex semantics and side-effects.
To alleviate the problem, \citet{mimic} introduced a technique
to infer JavaScript code for opaque code using concrete execution.
They leveraged ES6 \jscode{Proxy} objects to collect partial execution traces
from opaque code and synthesized JavaScript code using the extracted behaviors.
Instead of synthesizing JavaScript code,
\citet{sra} presented a \textit{Sample-Run-Abstract (SRA)} approach for
on-demand modeling focusing on the current abstract states during static analysis
by sampling well-distributed concrete states.
However, all the previous work sacrifice the soundness of static analysis.
On the contrary, while dynamic shortcuts is not always applicable to opaque
functions, it is sound if it is applicable.

\paragraph{Pruning Analysis Scope}
Another approach to utilize dynamic analysis for JavaScript static analysis is
to prune the scope of analysis.  \citet{determinacy} proposed dynamic
determinacy analysis.  They specialized target source code with determinacy
facts so that static analysis can get benefits from elimination of \jscode{eval}
and constant property names.  \citet{blendedJS} introduced \textit{blended taint
analysis}, which specializes JavaScript dynamic
language features such as dynamic code generation or variadic function calls.
It first performs dynamic analysis to collect traces with concrete values used
in dynamic language features and restricts the semantics of features based on the
collected traces during static analysis.  \citet{battles, eha} utilize three points to reduce
analysis scope: initial states, dynamically loaded files, and event handlers.
Unfortunately, all the above approaches except~\cite{determinacy} do not preserve
soundness of static analysis unlike our approach using dynamic shortcuts.

%% file: conclusion.tex
\section{Conclusion}\label{sec:conclusion}
We presented a novel technique for JavaScript static analysis using
\textit{dynamic shortcuts}.  It can significantly accelerate static analysis
and lessen the modeling efforts for opaque code by freely leveraging high
performance of dynamic analysis for concretely executable program parts.
To maximize such benefits, we proposed \textit{{\sealed} execution}, which
performs concrete execution using {\sealed} values for abstract values.  We
formally defined static analysis using dynamic shortcuts in the abstract
interpretation framework and proved its soundness and termination.  We developed
$\tool$ as a prototype implementation of the proposed approach by extending a
combination of the state-of-the-art static and dynamic analyzers SAFE and
Jalangi.  Our tool accelerates the speed of static analysis 22.30$\x$ for
original tests and 7.81$\x$ for abstracted tests of Lodash 4 library.  Moreover,
it reduces the number of failed assertions by 12.31\% by using {\sealed}
execution instead of manual modeling for 22 opaque functions on average.

%% file: main.bbl

\begin{thebibliography}{43}


\ifx \showCODEN    \undefined \def \showCODEN     #1{\unskip}     \fi
\ifx \showDOI      \undefined \def \showDOI       #1{#1}\fi
\ifx \showISBNx    \undefined \def \showISBNx     #1{\unskip}     \fi
\ifx \showISBNxiii \undefined \def \showISBNxiii  #1{\unskip}     \fi
\ifx \showISSN     \undefined \def \showISSN      #1{\unskip}     \fi
\ifx \showLCCN     \undefined \def \showLCCN      #1{\unskip}     \fi
\ifx \shownote     \undefined \def \shownote      #1{#1}          \fi
\ifx \showarticletitle \undefined \def \showarticletitle #1{#1}   \fi
\ifx \showURL      \undefined \def \showURL       {\relax}        \fi
\providecommand\bibfield[2]{#2}
\providecommand\bibinfo[2]{#2}
\providecommand\natexlab[1]{#1}
\providecommand\showeprint[2][]{arXiv:#2}

\bibitem[\protect\citeauthoryear{??}{ele}{2020}]%
        {electron}
 \bibinfo{year}{2020}\natexlab{}.
\newblock \bibinfo{booktitle}{\emph{{Electron - A framework for cross-platform
  desktop apps with JavaScript, HTML, and CSS}}}.
\newblock
\urldef\tempurl%
\url{https://www.electronjs.org/}
\showURL{%
Retrieved May 25, 2021 from \tempurl}


\bibitem[\protect\citeauthoryear{??}{esp}{2020}]%
        {espruino}
 \bibinfo{year}{2020}\natexlab{}.
\newblock \bibinfo{booktitle}{\emph{{Espruino - An open-source JavaScript
  interpreter for microcontrollers}}}.
\newblock
\urldef\tempurl%
\url{https://www.espruino.com/}
\showURL{%
Retrieved May 25, 2021 from \tempurl}


\bibitem[\protect\citeauthoryear{??}{lod}{2020}]%
        {lodash}
 \bibinfo{year}{2020}\natexlab{}.
\newblock \bibinfo{booktitle}{\emph{{Lodash - A modern JavaScript library
  delivering modularity, performance, and extras}}}.
\newblock
\urldef\tempurl%
\url{https://lodash.com/}
\showURL{%
Retrieved May 25, 2021 from \tempurl}


\bibitem[\protect\citeauthoryear{??}{mod}{2020}]%
        {moddable}
 \bibinfo{year}{2020}\natexlab{}.
\newblock \bibinfo{booktitle}{\emph{{Moddable - Tools to create open IoT
  products using standard JavaScript on low cast microcontrollers}}}.
\newblock
\urldef\tempurl%
\url{https://www.moddable.com/}
\showURL{%
Retrieved May 25, 2021 from \tempurl}


\bibitem[\protect\citeauthoryear{??}{nod}{2020}]%
        {nodejs}
 \bibinfo{year}{2020}\natexlab{}.
\newblock \bibinfo{booktitle}{\emph{{Node.js - A JavaScript runtime built on
  Chrome's V8 JavaScript engine}}}.
\newblock
\urldef\tempurl%
\url{https://nodejs.org/}
\showURL{%
Retrieved May 25, 2021 from \tempurl}


\bibitem[\protect\citeauthoryear{??}{rea}{2020}]%
        {react-native}
 \bibinfo{year}{2020}\natexlab{}.
\newblock \bibinfo{booktitle}{\emph{{React Native - A framework for building
  native apps using React}}}.
\newblock
\urldef\tempurl%
\url{https://reactnative.dev/}
\showURL{%
Retrieved May 25, 2021 from \tempurl}


\bibitem[\protect\citeauthoryear{??}{es6}{2020}]%
        {es6}
 \bibinfo{year}{2020}\natexlab{}.
\newblock \bibinfo{booktitle}{\emph{{Standard ECMA-262 6th Edition, ECMAScript
  2015 Language Specification}}}.
\newblock
\urldef\tempurl%
\url{https://262.ecma-international.org/6.0/}
\showURL{%
Retrieved May 25, 2021 from \tempurl}


\bibitem[\protect\citeauthoryear{??}{sun}{2020}]%
        {sunspider}
 \bibinfo{year}{2020}\natexlab{}.
\newblock \bibinfo{booktitle}{\emph{{SunSpider Javascript Benchmark}}}.
\newblock
\urldef\tempurl%
\url{https://webkit.org/perf/sunspider/sunspider.html}
\showURL{%
Retrieved May 25, 2021 from \tempurl}


\bibitem[\protect\citeauthoryear{??}{zoo}{2020}]%
        {zoom}
 \bibinfo{year}{2020}\natexlab{}.
\newblock \bibinfo{booktitle}{\emph{{Zoom - A videotelephony software program
  developed by Zoom Video Communications}}}.
\newblock
\urldef\tempurl%
\url{https://zoom.us/}
\showURL{%
Retrieved May 25, 2021 from \tempurl}


\bibitem[\protect\citeauthoryear{Amadini, Jordan, Gange, Gauthier, Schachte,
  S{\o}ndergaard, Stuckey, and Zhang}{Amadini et~al\mbox{.}}{2017}]%
        {combining-string}
\bibfield{author}{\bibinfo{person}{Roberto Amadini}, \bibinfo{person}{Alexander
  Jordan}, \bibinfo{person}{Graeme Gange}, \bibinfo{person}{Fran{\c{c}}ois
  Gauthier}, \bibinfo{person}{Peter Schachte}, \bibinfo{person}{Harald
  S{\o}ndergaard}, \bibinfo{person}{Peter~J Stuckey}, {and}
  \bibinfo{person}{Chenyi Zhang}.} \bibinfo{year}{2017}\natexlab{}.
\newblock \showarticletitle{{Combining String Abstract Domains for JavaScript
  Analysis: An Evaluation}}. In \bibinfo{booktitle}{\emph{Proceedings of the
  23rd International Conference on Tools and Algorithms for the Construction
  and Analysis of Systems (TACAS)}}.
\newblock
\urldef\tempurl%
\url{https://doi.org/10.1007/978-3-662-54577-5_3}
\showDOI{\tempurl}


\bibitem[\protect\citeauthoryear{Bae, Cho, Lim, and Ryu}{Bae
  et~al\mbox{.}}{2014}]%
        {safewapi}
\bibfield{author}{\bibinfo{person}{SungGyeong Bae}, \bibinfo{person}{Hyunghun
  Cho}, \bibinfo{person}{Inho Lim}, {and} \bibinfo{person}{Sukyoung Ryu}.}
  \bibinfo{year}{2014}\natexlab{}.
\newblock \showarticletitle{{SAFEWAPI: Web API Misuse Detector for Web
  Applications}}. In \bibinfo{booktitle}{\emph{Proceedings of the 22nd ACM
  SIGSOFT International Symposium on Foundations of Software Engineering
  (FSE)}}.
\newblock
\urldef\tempurl%
\url{https://doi.org/10.1145/2635868.2635916}
\showDOI{\tempurl}


\bibitem[\protect\citeauthoryear{Cadar, Dunbar, and Engler}{Cadar
  et~al\mbox{.}}{2008}]%
        {klee}
\bibfield{author}{\bibinfo{person}{Cristian Cadar}, \bibinfo{person}{Daniel
  Dunbar}, {and} \bibinfo{person}{Dawson Engler}.}
  \bibinfo{year}{2008}\natexlab{}.
\newblock \showarticletitle{{KLEE: Unassisted and Automatic Generation of
  High-Coverage Tests for Complex Systems Programs}}. In
  \bibinfo{booktitle}{\emph{Proceedings of the 8th USENIX Symposium on
  Operating Systems Design and Implementation (OSDI)}},
  Vol.~\bibinfo{volume}{8}. \bibinfo{pages}{209--224}.
\newblock
\urldef\tempurl%
\url{https://dl.acm.org/doi/10.5555/1855741.1855756}
\showURL{%
\tempurl}


\bibitem[\protect\citeauthoryear{Cousot and Cousot}{Cousot and Cousot}{1977}]%
        {abs-interp-1977}
\bibfield{author}{\bibinfo{person}{Patrick Cousot} {and}
  \bibinfo{person}{Radhia Cousot}.} \bibinfo{year}{1977}\natexlab{}.
\newblock \showarticletitle{{Abstract Interpretation: A Unified Lattice Model
  for Static Analysis of Programs by Construction or Approximation of
  Fixpoints}}. In \bibinfo{booktitle}{\emph{Proceedings of the 4th ACM
  SIGACT-SIGPLAN Symposium on Principles of Programming languages (POPL)}}.
\newblock
\urldef\tempurl%
\url{https://doi.org/10.1145/512950.512973}
\showDOI{\tempurl}


\bibitem[\protect\citeauthoryear{Cousot and Cousot}{Cousot and Cousot}{1992}]%
        {abs-interp-1992}
\bibfield{author}{\bibinfo{person}{Patrick Cousot} {and}
  \bibinfo{person}{Radhia Cousot}.} \bibinfo{year}{1992}\natexlab{}.
\newblock \showarticletitle{{Abstract interpretation frameworks}}.
\newblock \bibinfo{journal}{\emph{Journal of Logic and Computation (JLC)}}
  \bibinfo{volume}{2}, \bibinfo{number}{4} (\bibinfo{year}{1992}),
  \bibinfo{pages}{511--547}.
\newblock
\urldef\tempurl%
\url{https://doi.org/10.1093/logcom/2.4.511}
\showDOI{\tempurl}


\bibitem[\protect\citeauthoryear{Godefroid, Klarlund, and Sen}{Godefroid
  et~al\mbox{.}}{2005}]%
        {dart}
\bibfield{author}{\bibinfo{person}{Patrice Godefroid}, \bibinfo{person}{Nils
  Klarlund}, {and} \bibinfo{person}{Koushik Sen}.}
  \bibinfo{year}{2005}\natexlab{}.
\newblock \showarticletitle{{DART: Directed automated random testing}}. In
  \bibinfo{booktitle}{\emph{Proceedings of the ACM SIGPLAN conference on
  Programming language design and implementation (PLDI)}}.
\newblock
\urldef\tempurl%
\url{https://doi.org/10.1145/1065010.1065036}
\showDOI{\tempurl}


\bibitem[\protect\citeauthoryear{Godefroid, Levin, and Molnar}{Godefroid
  et~al\mbox{.}}{2012}]%
        {sage}
\bibfield{author}{\bibinfo{person}{Patrice Godefroid},
  \bibinfo{person}{Michael~Y Levin}, {and} \bibinfo{person}{David Molnar}.}
  \bibinfo{year}{2012}\natexlab{}.
\newblock \showarticletitle{{SAGE: Whitebox Fuzzing for Security Testing}}.
\newblock \bibinfo{journal}{\emph{Communications of the ACM (CACM)}}
  \bibinfo{volume}{55}, \bibinfo{number}{3} (\bibinfo{year}{2012}),
  \bibinfo{pages}{40--44}.
\newblock
\urldef\tempurl%
\url{https://doi.org/10.1145/2093548.2093564}
\showDOI{\tempurl}


\bibitem[\protect\citeauthoryear{Gong, Pradel, Sridharan, and Sen}{Gong
  et~al\mbox{.}}{2015}]%
        {dlint}
\bibfield{author}{\bibinfo{person}{Liang Gong}, \bibinfo{person}{Michael
  Pradel}, \bibinfo{person}{Manu Sridharan}, {and} \bibinfo{person}{Koushik
  Sen}.} \bibinfo{year}{2015}\natexlab{}.
\newblock \showarticletitle{{DLint: Dynamically Checking Bad Coding Practices
  in JavaScript}}. In \bibinfo{booktitle}{\emph{Proceedings of the 24th
  International Symposium on Software Testing and Analysis (ISSTA)}}.
\newblock
\urldef\tempurl%
\url{https://doi.org/10.1145/2771783.2771809}
\showDOI{\tempurl}


\bibitem[\protect\citeauthoryear{Heule, Sridharan, and Chandra}{Heule
  et~al\mbox{.}}{2015}]%
        {mimic}
\bibfield{author}{\bibinfo{person}{Stefan Heule}, \bibinfo{person}{Manu
  Sridharan}, {and} \bibinfo{person}{Satish Chandra}.}
  \bibinfo{year}{2015}\natexlab{}.
\newblock \showarticletitle{{Mimic: Computing Models for Opaque Code}}. In
  \bibinfo{booktitle}{\emph{Proceedings of the 10th Joint Meeting of the
  European Software Engineering Conference and the ACM SIGSOFT Symposium on the
  Foundations of Software Engineering (ESEC/FSE)}}.
\newblock
\urldef\tempurl%
\url{https://doi.org/10.1145/2786805.2786875}
\showDOI{\tempurl}


\bibitem[\protect\citeauthoryear{Jensen, M{\o}ller, and Thiemann}{Jensen
  et~al\mbox{.}}{2009}]%
        {tajs}
\bibfield{author}{\bibinfo{person}{Simon~Holm Jensen}, \bibinfo{person}{Anders
  M{\o}ller}, {and} \bibinfo{person}{Peter Thiemann}.}
  \bibinfo{year}{2009}\natexlab{}.
\newblock \showarticletitle{{Type Analysis for JavaScript}}. In
  \bibinfo{booktitle}{\emph{Proceedings of the 16th International Symposium on
  Static Analysis (SAS)}}.
\newblock
\urldef\tempurl%
\url{https://doi.org/10.1007/978-3-642-03237-0_17}
\showDOI{\tempurl}


\bibitem[\protect\citeauthoryear{Kashyap, Dewey, Kuefner, Wagner, Gibbons,
  Sarracino, Wiedermann, and Hardekopf}{Kashyap et~al\mbox{.}}{2014}]%
        {jsai}
\bibfield{author}{\bibinfo{person}{Vineeth Kashyap}, \bibinfo{person}{Kyle
  Dewey}, \bibinfo{person}{Ethan~A. Kuefner}, \bibinfo{person}{John Wagner},
  \bibinfo{person}{Kevin Gibbons}, \bibinfo{person}{John Sarracino},
  \bibinfo{person}{Ben Wiedermann}, {and} \bibinfo{person}{Ben Hardekopf}.}
  \bibinfo{year}{2014}\natexlab{}.
\newblock \showarticletitle{JSAI: A Static Analysis Platform for JavaScript}.
  In \bibinfo{booktitle}{\emph{Proceedings of the 22nd ACM SIGSOFT
  International Symposium on Foundations of Software Engineering (FSE)}}.
\newblock
\urldef\tempurl%
\url{https://doi.org/10.1145/2635868.2635904}
\showDOI{\tempurl}


\bibitem[\protect\citeauthoryear{Kim, Rival, and Ryu}{Kim
  et~al\mbox{.}}{2018}]%
        {sens-toplas}
\bibfield{author}{\bibinfo{person}{Se-Won Kim}, \bibinfo{person}{Xavier Rival},
  {and} \bibinfo{person}{Sukyoung Ryu}.} \bibinfo{year}{2018}\natexlab{}.
\newblock \showarticletitle{{A Theoretical Foundation of Sensitivity in an
  Abstract Interpretation Framework}}.
\newblock \bibinfo{journal}{\emph{ACM Transactions on Programming Languages and
  Systems (TOPLAS)}} \bibinfo{volume}{40}, \bibinfo{number}{3}
  (\bibinfo{year}{2018}), \bibinfo{pages}{1--44}.
\newblock
\urldef\tempurl%
\url{https://doi.org/10.1145/3230624}
\showDOI{\tempurl}


\bibitem[\protect\citeauthoryear{King}{King}{1976}]%
        {symbolic}
\bibfield{author}{\bibinfo{person}{James~C King}.}
  \bibinfo{year}{1976}\natexlab{}.
\newblock \showarticletitle{{Symbolic execution and program testing}}.
\newblock \bibinfo{journal}{\emph{Communications of the ACM (CACM)}}
  \bibinfo{volume}{19}, \bibinfo{number}{7} (\bibinfo{year}{1976}),
  \bibinfo{pages}{385--394}.
\newblock
\urldef\tempurl%
\url{https://doi.org/10.1145/360248.360252}
\showDOI{\tempurl}


\bibitem[\protect\citeauthoryear{Ko, Rival, and Ryu}{Ko et~al\mbox{.}}{2017}]%
        {weaklyAPLAS}
\bibfield{author}{\bibinfo{person}{Yoonseok Ko}, \bibinfo{person}{Xavier
  Rival}, {and} \bibinfo{person}{Sukyoung Ryu}.}
  \bibinfo{year}{2017}\natexlab{}.
\newblock \showarticletitle{{Weakly Sensitive Analysis for Unbounded Iteration
  over JavaScript Objects}}. In \bibinfo{booktitle}{\emph{Proceedings of the
  15th Asian Symposium on Programming Languages and Systems (APLAS)}}.
\newblock
\urldef\tempurl%
\url{https://doi.org/10.1007/978-3-319-71237-6_8}
\showDOI{\tempurl}


\bibitem[\protect\citeauthoryear{Ko, Rival, and Ryu}{Ko et~al\mbox{.}}{2019}]%
        {weaklySPE}
\bibfield{author}{\bibinfo{person}{Yoonseok Ko}, \bibinfo{person}{Xavier
  Rival}, {and} \bibinfo{person}{Sukyoung Ryu}.}
  \bibinfo{year}{2019}\natexlab{}.
\newblock \showarticletitle{{Weakly sensitive analysis for JavaScript
  object-manipulating programs}}.
\newblock \bibinfo{journal}{\emph{Software: Practice and Experience (SPE)}}
  \bibinfo{volume}{49}, \bibinfo{number}{5} (\bibinfo{year}{2019}),
  \bibinfo{pages}{840--884}.
\newblock
\urldef\tempurl%
\url{https://doi.org/10.1002/spe.2676}
\showDOI{\tempurl}


\bibitem[\protect\citeauthoryear{Lee, Won, Jin, Cho, and Ryu}{Lee
  et~al\mbox{.}}{2012}]%
        {safe}
\bibfield{author}{\bibinfo{person}{Hongki Lee}, \bibinfo{person}{Sooncheol
  Won}, \bibinfo{person}{Joonho Jin}, \bibinfo{person}{Junhee Cho}, {and}
  \bibinfo{person}{Sukyoung Ryu}.} \bibinfo{year}{2012}\natexlab{}.
\newblock \showarticletitle{{SAFE: Formal Specification and Implementation of a
  Scalable Analysis Framework for ECMAScript}}. In
  \bibinfo{booktitle}{\emph{Proceedings of 19th International Workshop on
  Foundations of Object-Oriented Languages (FOOL)}}.
\newblock


\bibitem[\protect\citeauthoryear{Madsen and Andreasen}{Madsen and
  Andreasen}{2014}]%
        {string}
\bibfield{author}{\bibinfo{person}{Magnus Madsen} {and} \bibinfo{person}{Esben
  Andreasen}.} \bibinfo{year}{2014}\natexlab{}.
\newblock \showarticletitle{{String Analysis for Dynamic Field Access}}. In
  \bibinfo{booktitle}{\emph{Proceedings of the 23rd International Conference on
  Compiler Construction (CC)}}.
\newblock
\urldef\tempurl%
\url{https://doi.org/10.1007/978-3-642-54807-9_12}
\showDOI{\tempurl}


\bibitem[\protect\citeauthoryear{Might and Shivers}{Might and Shivers}{2006}]%
        {abstract-gc-counting}
\bibfield{author}{\bibinfo{person}{Matthew Might} {and} \bibinfo{person}{Olin
  Shivers}.} \bibinfo{year}{2006}\natexlab{}.
\newblock \showarticletitle{{Improving Flow Analyses via $\Gamma$CFA: Abstract
  Garbage Collection and Counting}}. In \bibinfo{booktitle}{\emph{Proceedings
  of the 11th ACM SIGPLAN International Conference on Functional Programming
  (ICFP)}}.
\newblock
\urldef\tempurl%
\url{https://doi.org/10.1145/1159803.1159807}
\showDOI{\tempurl}


\bibitem[\protect\citeauthoryear{Nielsen and M{\o}ller}{Nielsen and
  M{\o}ller}{2020}]%
        {value-partitioning}
\bibfield{author}{\bibinfo{person}{Benjamin~Barslev Nielsen} {and}
  \bibinfo{person}{Anders M{\o}ller}.} \bibinfo{year}{2020}\natexlab{}.
\newblock \showarticletitle{{Value Partitioning: A Lightweight Approach to
  Relational Static Analysis for JavaScript}}. In
  \bibinfo{booktitle}{\emph{Proceedings of the 34th European Conference on
  Object-Oriented Programming (ECOOP)}}.
\newblock
\urldef\tempurl%
\url{https://doi.org/10.4230/LIPIcs.ECOOP.2020.16}
\showDOI{\tempurl}


\bibitem[\protect\citeauthoryear{Park, Im, and Ryu}{Park
  et~al\mbox{.}}{2016a}]%
        {regex}
\bibfield{author}{\bibinfo{person}{Changhee Park}, \bibinfo{person}{Hyeonseung
  Im}, {and} \bibinfo{person}{Sukyoung Ryu}.} \bibinfo{year}{2016}\natexlab{a}.
\newblock \showarticletitle{{Precise and Scalable Static Analysis of jQuery
  using a Regular Expression Domain}}. In \bibinfo{booktitle}{\emph{Proceedings
  of the 12th Symposium on Dynamic Languages (DLS)}}.
\newblock
\urldef\tempurl%
\url{https://doi.org/10.1145/2989225.2989228}
\showDOI{\tempurl}


\bibitem[\protect\citeauthoryear{Park, Lee, and Ryu}{Park
  et~al\mbox{.}}{2018a}]%
        {lsaSPE}
\bibfield{author}{\bibinfo{person}{Changhee Park}, \bibinfo{person}{Hongki
  Lee}, {and} \bibinfo{person}{Sukyoung Ryu}.}
  \bibinfo{year}{2018}\natexlab{a}.
\newblock \showarticletitle{{Static analysis of JavaScript libraries in a
  scalable and precise way using loop sensitivity}}.
\newblock \bibinfo{journal}{\emph{Software: Practice and Experience (SPE)}}
  \bibinfo{volume}{48}, \bibinfo{number}{4} (\bibinfo{year}{2018}),
  \bibinfo{pages}{911--944}.
\newblock
\urldef\tempurl%
\url{https://doi.org/10.1002/spe.2676}
\showDOI{\tempurl}


\bibitem[\protect\citeauthoryear{Park and Ryu}{Park and Ryu}{2015}]%
        {lsaECOOP}
\bibfield{author}{\bibinfo{person}{Changhee Park} {and}
  \bibinfo{person}{Sukyoung Ryu}.} \bibinfo{year}{2015}\natexlab{}.
\newblock \showarticletitle{{Scalable and Precise Static Analysis of JavaScript
  Applications via Loop-Sensitivity}}. In \bibinfo{booktitle}{\emph{Proceedings
  of the 29th European Conference on Object-Oriented Programming (ECOOP)}}.
\newblock
\urldef\tempurl%
\url{https://doi.org/10.4230/LIPIcs.ECOOP.2015.735}
\showDOI{\tempurl}


\bibitem[\protect\citeauthoryear{Park}{Park}{2014}]%
        {safets}
\bibfield{author}{\bibinfo{person}{Jihyeok Park}.}
  \bibinfo{year}{2014}\natexlab{}.
\newblock \showarticletitle{{JavaScript API misuse detection by using
  typescript}}. In \bibinfo{booktitle}{\emph{Proceedings of the companion
  publication of the 13th international conference on Modularity}}.
\newblock
\urldef\tempurl%
\url{https://doi.org/10.1145/2584469.2584472}
\showDOI{\tempurl}


\bibitem[\protect\citeauthoryear{Park, Jordan, and Ryu}{Park
  et~al\mbox{.}}{2019}]%
        {sra}
\bibfield{author}{\bibinfo{person}{Joonyoung Park}, \bibinfo{person}{Alexander
  Jordan}, {and} \bibinfo{person}{Sukyoung Ryu}.}
  \bibinfo{year}{2019}\natexlab{}.
\newblock \showarticletitle{{Automatic Modeling of Opaque Code for JavaScript
  Static Analysis}}. In \bibinfo{booktitle}{\emph{Proceedings of the 22nd
  International Conference on Fundamental Approaches to Software Engineering
  (FASE)}}.
\newblock
\urldef\tempurl%
\url{https://doi.org/10.1007/978-3-030-16722-6_3}
\showDOI{\tempurl}


\bibitem[\protect\citeauthoryear{Park, Lim, and Ryu}{Park
  et~al\mbox{.}}{2016b}]%
        {battles}
\bibfield{author}{\bibinfo{person}{Joonyoung Park}, \bibinfo{person}{Inho Lim},
  {and} \bibinfo{person}{Sukyoung Ryu}.} \bibinfo{year}{2016}\natexlab{b}.
\newblock \showarticletitle{{Battles with False Positives in Static Analysis of
  JavaScript Web Applications in the Wild}}. In
  \bibinfo{booktitle}{\emph{Proceedings of the 38th IEEE/ACM International
  Conference on Software Engineering Companion (ICSE-C)}}.
\newblock
\urldef\tempurl%
\url{https://doi.org/10.1145/2889160.2889227}
\showDOI{\tempurl}


\bibitem[\protect\citeauthoryear{Park, Rival, and Ryu}{Park
  et~al\mbox{.}}{2017a}]%
        {revisit-recency}
\bibfield{author}{\bibinfo{person}{Jihyeok Park}, \bibinfo{person}{Xavier
  Rival}, {and} \bibinfo{person}{Sukyoung Ryu}.}
  \bibinfo{year}{2017}\natexlab{a}.
\newblock \showarticletitle{{Revisiting Recency Abstraction for JavaScript:
  Towards an Intuitive, Compositional, and Efficient Heap Abstraction}}. In
  \bibinfo{booktitle}{\emph{Proceedings of the 6th ACM SIGPLAN International
  Workshop on State Of the Art in Program Analysis (SOAP)}}.
\newblock
\urldef\tempurl%
\url{https://doi.org/10.1145/3088515.3088516}
\showDOI{\tempurl}


\bibitem[\protect\citeauthoryear{Park, Ryou, Park, and Ryu}{Park
  et~al\mbox{.}}{2017b}]%
        {safe2}
\bibfield{author}{\bibinfo{person}{Jihyeok Park}, \bibinfo{person}{Yeonhee
  Ryou}, \bibinfo{person}{Joonyoung Park}, {and} \bibinfo{person}{Sukyoung
  Ryu}.} \bibinfo{year}{2017}\natexlab{b}.
\newblock \showarticletitle{{Analysis of JavaScript Web Applications Using SAFE
  2.0}}. In \bibinfo{booktitle}{\emph{Proceedings of the 39th IEEE/ACM
  International Conference on Software Engineering Companion (ICSE-C)}}.
\newblock
\urldef\tempurl%
\url{https://doi.org/10.1109/ICSE-C.2017.4}
\showDOI{\tempurl}


\bibitem[\protect\citeauthoryear{Park, Sun, and Ryu}{Park
  et~al\mbox{.}}{2018b}]%
        {eha}
\bibfield{author}{\bibinfo{person}{Joonyoung Park}, \bibinfo{person}{Kwangwon
  Sun}, {and} \bibinfo{person}{Sukyoung Ryu}.}
  \bibinfo{year}{2018}\natexlab{b}.
\newblock \showarticletitle{{EventHandler-Based Analysis Framework for Web Apps
  Using Dynamically Collected States}}. In
  \bibinfo{booktitle}{\emph{Proceedings of the 21st International Conference on
  Fundamental Approaches to Software Engineering (FASE)}}.
\newblock
\urldef\tempurl%
\url{https://doi.org/10.1007/978-3-319-89363-1_8}
\showDOI{\tempurl}


\bibitem[\protect\citeauthoryear{Sch{\"a}fer, Sridharan, Dolby, and
  Tip}{Sch{\"a}fer et~al\mbox{.}}{2013}]%
        {determinacy}
\bibfield{author}{\bibinfo{person}{Max Sch{\"a}fer}, \bibinfo{person}{Manu
  Sridharan}, \bibinfo{person}{Julian Dolby}, {and} \bibinfo{person}{Frank
  Tip}.} \bibinfo{year}{2013}\natexlab{}.
\newblock \showarticletitle{{Dynamic Determinacy Analysis}}. In
  \bibinfo{booktitle}{\emph{Proceedings of the 34th annual ACM SIGPLAN
  conference on Programming Language Design and Implementation (PLDI)}}.
\newblock
\urldef\tempurl%
\url{https://doi.org/10.1145/2499370.2462168}
\showDOI{\tempurl}


\bibitem[\protect\citeauthoryear{Sen, Kalasapur, Brutch, and Gibbs}{Sen
  et~al\mbox{.}}{2013}]%
        {jalangi}
\bibfield{author}{\bibinfo{person}{Koushik Sen}, \bibinfo{person}{Swaroop
  Kalasapur}, \bibinfo{person}{Tasneem Brutch}, {and} \bibinfo{person}{Simon
  Gibbs}.} \bibinfo{year}{2013}\natexlab{}.
\newblock \showarticletitle{{Jalangi: A Selective Record-Replay and Dynamic
  Analysis Framework for JavaScript}}. In \bibinfo{booktitle}{\emph{Proceedings
  of the 10th Joint Meeting of the European Software Engineering Conference and
  the ACM SIGSOFT Symposium on the Foundations of Software Engineering
  (ESEC/FSE)}}.
\newblock
\urldef\tempurl%
\url{https://doi.org/10.1145/2491411.2491447}
\showDOI{\tempurl}


\bibitem[\protect\citeauthoryear{Sridharan, Dolby, Chandra, Sch{\"a}fer, and
  Tip}{Sridharan et~al\mbox{.}}{2012}]%
        {correlation}
\bibfield{author}{\bibinfo{person}{Manu Sridharan}, \bibinfo{person}{Julian
  Dolby}, \bibinfo{person}{Satish Chandra}, \bibinfo{person}{Max Sch{\"a}fer},
  {and} \bibinfo{person}{Frank Tip}.} \bibinfo{year}{2012}\natexlab{}.
\newblock \showarticletitle{{Correlation Tracking for Points-To Analysis of
  JavaScript}}. In \bibinfo{booktitle}{\emph{Proceedings of the 26th European
  Conference on Object-Oriented Programming (ECOOP)}}.
\newblock
\urldef\tempurl%
\url{https://doi.org/10.1007/978-3-642-31057-7_20}
\showDOI{\tempurl}


\bibitem[\protect\citeauthoryear{Stein, Nielsen, Chang, and M{\o}ller}{Stein
  et~al\mbox{.}}{2019}]%
        {value-refinement}
\bibfield{author}{\bibinfo{person}{Benno Stein},
  \bibinfo{person}{Benjamin~Barslev Nielsen}, \bibinfo{person}{Bor-Yuh~Evan
  Chang}, {and} \bibinfo{person}{Anders M{\o}ller}.}
  \bibinfo{year}{2019}\natexlab{}.
\newblock \showarticletitle{{Static Analysis with Demand-Driven Value
  Refinement}}. In \bibinfo{booktitle}{\emph{Proceedings of the 34th ACM
  SIGPLAN conference on Object-Oriented Programming, Systems, Languages, and
  Applications (OOPSLA)}}.
\newblock
\urldef\tempurl%
\url{https://doi.org/10.1145/3360566}
\showDOI{\tempurl}


\bibitem[\protect\citeauthoryear{Toman and Grossman}{Toman and
  Grossman}{2019}]%
        {concerto}
\bibfield{author}{\bibinfo{person}{John Toman} {and} \bibinfo{person}{Dan
  Grossman}.} \bibinfo{year}{2019}\natexlab{}.
\newblock \showarticletitle{{Concerto: A Framework for Combined Concrete and
  Abstract Interpretation}}. In \bibinfo{booktitle}{\emph{Proceedings of the
  46th ACM SIGPLAN Symposium on Principles of Programming Languages (POPL)}}.
\newblock
\urldef\tempurl%
\url{https://doi.org/10.1145/3290356}
\showDOI{\tempurl}


\bibitem[\protect\citeauthoryear{Wei and Ryder}{Wei and Ryder}{2013}]%
        {blendedJS}
\bibfield{author}{\bibinfo{person}{Shiyi Wei} {and} \bibinfo{person}{Barbara~G
  Ryder}.} \bibinfo{year}{2013}\natexlab{}.
\newblock \showarticletitle{{Practical Blended Taint Analysis for JavaScript}}.
  In \bibinfo{booktitle}{\emph{Proceedings of the 22th International Symposium
  on Software Testing and Analysis (ISSTA)}}.
\newblock
\urldef\tempurl%
\url{https://doi.org/10.1145/2483760.2483788}
\showDOI{\tempurl}


\end{thebibliography}
